\documentclass[a4paper]{article}
\usepackage[margin=1in]{geometry}
\usepackage{authblk}
\usepackage[utf8]{inputenc}
\usepackage[T1]{fontenc}
\usepackage{microtype}
\usepackage{amsfonts}
\usepackage{amssymb}
\usepackage{amsthm}
\usepackage{amsmath}
\usepackage{comment}
\usepackage{times}
\usepackage{enumitem}
\usepackage[hidelinks]{hyperref}
\usepackage[capitalise]{cleveref}
\usepackage{thmtools}

\newtheorem{theorem}{Theorem}[section]
\newtheorem{lemma}[theorem]{Lemma}
\newtheorem{construction}[theorem]{Construction}
\newtheorem{fact}[theorem]{Fact}
\newtheorem{observation}[theorem]{Observation}

\newtheorem{definition}[theorem]{Definition}
\newtheorem{proposition}[theorem]{Proposition}
\theoremstyle{remark}

\includecomment{full}
\excludecomment{camera}

\newcommand\bigO{\mathcal{O}}
\newcommand{\SA}{\text{\rm SA}}
\newcommand{\LCP}{\text{\rm LCP}}
\newcommand{\BWT}{\text{\rm BWT}}
\newcommand{\lcp}{\text{\rm lcp}}
\newcommand{\per}{\text{\rm per}}
\newcommand{\type}{\text{\rm type}}
\newcommand{\pos}{\text{\rm pos}}
\newcommand{\Lroot}{\text{\rm L-root}}
\newcommand{\Lexp}{\text{\rm L-exp}}
\newcommand{\revbits}[1]{\overline{#1}}
\newcommand{\id}{\mathsf{id}}
\newcommand{\dotdot}{\mathinner{\ldotp\ldotp}}
\newcommand{\dd}{\dotdot}
\newcommand{\LCE}{\operatorname{LCE}}
\newcommand{\PP}{\mathcal{P}}
\renewcommand{\P}{\mathsf{P}}
\newcommand{\A}{\mathsf{A}}
\renewcommand{\Pr}{\mathbb{P}}
\newcommand{\F}{\mathcal{F}}
\newcommand{\D}{\mathcal{D}}
\newcommand{\BB}{\mathcal{B}}
\newcommand{\R}{\mathsf{R}}
\newcommand{\Q}{\mathsf{Q}}
\renewcommand{\S}{\mathsf{S}}
\newcommand{\B}{\mathsf{B}}
\newcommand{\ceil}[1]{\left\lceil #1 \right\rceil}
\newcommand{\floor}[1]{\left\lfloor #1 \right\rfloor}
\newcommand{\sub}{\subseteq}
\newcommand{\sm}{\setminus}
\newcommand{\eps}{\varepsilon}
\newcommand{\integ}{\mathbb{Z}}
\newcommand{\suc}{\mathrm{succ}}
\newcommand{\rank}{\mathrm{rank}}
\newcommand{\Exp}{\mathbb{E}}
\newcommand{\Oh}{\mathcal{O}}

\newcommand{\fullonly}{}

\crefname{construction}{Construction}{Constructions}
\crefname{fact}{Fact}{Facts}

\begin{document}

\setenumerate{itemsep=0ex, parsep=1pt, topsep=1pt}

\renewcommand\Affilfont{\normalsize}

\title{String Synchronizing Sets: Sublinear-Time BWT
  Construction\\ and Optimal LCE Data Structure}
\author[,1]{Dominik Kempa\thanks{Supported by the Centre for
Discrete Mathematics and its Applications (DIMAP) and by EPSRC award EP/N011163/1.}}
\author[,2,3]{Tomasz Kociumaka\thanks{Supported by ISF grants no. 824/17 and 1278/16 and by an ERC grant MPM under the EU's Horizon 2020 Research and Innovation Programme (grant no. 683064).}}

\affil[1]{Department of Computer Science, University of Warwick, United Kingdom}
\affil[ ]{\href{mailto:dominik.kempa@warwick.ac.uk}{\nolinkurl{dominik.kempa@warwick.ac.uk}}}
\affil[2]{Department of Computer Science, Bar-Ilan University, Ramat Gan, Israel}
\affil[3]{Institute of Informatics, University of Warsaw, Poland}
\affil[ ]{\href{mailto:kociumaka@mimuw.edu.pl}{\nolinkurl{kociumaka@mimuw.edu.pl}}}

\date{\vspace{-1.5cm}}
\maketitle

\begin{abstract}
  Burrows--Wheeler transform (BWT) is an invertible text transforma\-tion that, given a text $T$ of length $n$, permutes its symbols \mbox{according} to the lexicographic order of suffixes of $T$. BWT is one of~the~most heavily studied algorithms in data compression with numerous applications in indexing, sequence analysis, and bioinformatics.~Its construction is a bottleneck in many scenarios, and settling the complexity of this task is one of the most important unsolved \mbox{problems} in sequence analysis that has remained open for 25 years. Given a binary string of length $n$, occupying $\Oh(n/\log n)$ machine words, the BWT construction algorithm due to Hon et al. (SIAM~J.~Comput., 2009) runs in $\Oh(n)$ time and $\Oh(n/\log n)$ space. Recent \mbox{advancements} (Belazzougui, STOC 2014, and Munro et al., SODA~2017) focus on removing the alphabet-size dependency in the time complexity, but they still require $\Omega(n)$ time.  Despite the clearly suboptimal running time, the existing techniques appear to have reached their limits.
    
  In this paper, we propose the first algorithm that breaks the $\Oh(n)$-time barrier for BWT construction. Given a binary string of length $n$, our procedure builds the Burrows--Wheeler transform in $\Oh(n/\sqrt{\log n})$ time and $\Oh(n/\log n)$ space. We complement this result with a conditional lower bound proving that any further progress in the time complexity of BWT construction would yield faster algorithms for the very well studied problem of counting inversions: it would improve the state-of-the-art $\Oh(m\sqrt{\log m})$-time solution by Chan and P\v{a}tra\c{s}cu (SODA 2010). Our algorithm is based on a novel concept of string synchronizing sets, which is of independent interest. As one of the applications, we show that this technique lets us design a data structure of the optimal size $\Oh(n/\log n)$ that answers Longest Common Extension queries (LCE queries) in $\Oh(1)$ time and, furthermore, can be deterministically constructed in the optimal $\Oh(n/\log n)$ time.
  \end{abstract}

\section{Introduction}
\label{sec:intro}

The problem of text indexing is to preprocess an input text $T$ so
that given any query pattern $P$, we can quickly find the occurrences
of $P$ in $T$ (typically in $\bigO(|P|+{\rm occ})$ time, where $|P|$
is the length of $P$ and ${\rm occ}$ is the number of reported
occurrences).  Two classical data structures for this task are the
suffix tree~\cite{DBLP:conf/focs/Weiner73} and the suffix
array~\cite{DBLP:journals/siamcomp/ManberM93}.  The suffix tree is a
trie containing all suffixes of $T$ with each unary path compressed
into a single edge labeled by a text substring. The suffix array is a
list of suffixes of $T$ in the lexicographic order, with each suffix
encoded using its starting position.  Both data structures take
$\Theta(n)$ words of space, where $n$ is the length of $T$.  In addition to indexing, they underpin
dozens of applications in bioinformatics, data compression, and
information retrieval~\cite{DBLP:books/cu/Gusfield1997, myriad}. While the
suffix tree is slightly faster for some operations, the suffix array
is often preferred due to its simplicity and lower space usage.

Nowadays, however, indexing datasets of size close to the capacity of
available RAM is often required.  Even the suffix arrays are then
prohibitively large, particularly in applications where the text
consists of symbols from some alphabet $\Sigma$ of small size
$\sigma=|\Sigma|$ (e.g., $\Sigma=\{{\tt A}, {\tt C},
{\tt G}, {\tt T}\}$ and so $\sigma=4$ in bioinformatics).  For such collections, the
classical indexes are $\Theta(\log_{\sigma}n)$ times larger than the
text, which takes only $\Theta(n \log \sigma)$ bits, i.e.,
$\Theta(n/\log_{\sigma}n)$ machine words, and thus they prevent many
sequence analysis tasks to be performed without a significant penalty
in space consumption.

This situation changed dramatically in early 2000's, when
Ferragina and Manzini~\cite{fm2005}, as well as Grossi and
Vitter~\cite{GrossiV05}, independently proposed indexes with the
capabilities of the suffix array (incurring only a
$\bigO(\log^{\eps}n)$ slowdown in the query time) that take a space
asymptotically equal to that of the text (and with very small constant
factors). These indexes are known as the FM-index and the compressed
suffix array (CSA). The central component and the time and space
bottleneck in the construction of both the FM-index and
CSA\footnote{Although originally formulated in terms of the so-called
  ``$\Psi$ function''~\cite{GrossiV05}, it is now established
  (see, e.g.,~\cite{SODA2017,DBLP:journals/siamcomp/HonSS09}) that the CSA is essentially
  dual to the FM-index.} is the \emph{Burrows--Wheeler transform}
(BWT)~\cite{BWT}. BWT is an invertible permutation of the text that consists of
symbols preceding suffixes of text in the lexicographic order. Almost
immediately after their discovery, the BWT-based indexes replaced suffix
arrays and suffix trees and the BWT itself has become the basis of
almost all space-efficient algorithms for sequence analysis.  Modern
textbooks spend dozens of pages describing its
applications~\cite{ohl2013,MBCT2015,navarrobook}, and BWT-indexes are
widely used in practice; in bioinformatics, they
are the central component of many read-aligners~\cite{bowtie,BWA}.

\paragraph*{BWT Construction}
Given the practical importance of BWT,
its efficient construction emerged as one of the most important open
problems in the field of indexing and sequence analysis. The first breakthrough was the
algorithm of Hon et al.~\cite{DBLP:journals/siamcomp/HonSS09}, who reduced the time
complexity of BWT construction for binary strings from $\bigO(n \log
n)$ to $\bigO(n)$ time using working space of $\bigO(n)$ bits. This
bound has been recently generalized to any alphabet size
$\sigma$. More precisely, Belazzougui~\cite{STOC2014} described a
(randomized) $\bigO(n)$-time construction working in optimal space of
$\bigO(n/\log_{\sigma}n)$ words. Munro et al.~\cite{SODA2017} then proposed an alternative (and deterministic)
construction.  These
algorithms achieve the optimal construction space, but their running
time is still $\Omega(n)$, which is up to $\Theta(\log n)$ times more
than the lower bound of $\Omega(n/\log_{\sigma}n)$ time (required to
read the input and write the output).
Up until now, all $o(n)$-time algorithms required additional assumptions, such as
that the BWT is highly compressible using run-length encoding~\cite{DBLP:conf/soda/Kempa19}.

In this paper, we propose the first algorithm that always breaks the
$\bigO(n)$-time barrier for BWT construction. Given a binary string of
length $n$, our algorithm builds the Burrows--Wheeler transform in
$\bigO(n/\sqrt{\log n})$ time and $\bigO(n/\log n)$ space.  We
complement this result with a conditional lower bound proving that any
further progress in the time complexity of BWT construction would
imply faster algorithms for the very well studied problem of counting
inversions: it would improve the state-of-the-art 
$\bigO(m\sqrt{\log m})$-time solution by Chan and P\v{a}tra\c{s}cu~\cite{ChanP10}.  We
also generalize our construction to larger alphabets whose size $\sigma$ satisfies $\log \sigma \le \sqrt{\log n}$.
In this case, the running time is $\Oh(n \log \sigma / \sqrt{\log n})$
and the space complexity is $\Oh(n \log \sigma / \log n)$, proportional to the input and output size.

\paragraph{LCE Queries}
The \emph{Longest Common Extension} queries $\LCE(i,j)$ (also known as the Longest Common Prefix queries), given two 
positions in a text $T$, return the length of the longest common prefix
of the suffixes $T[i\dd n]$ and $T[j\dd n]$ starting at positions $i$ and $j$, respectively.
These queries were introduced by Landau and Vishkin~\cite{DBLP:journals/jcss/LandauV88}
in the context of approximate pattern matching.
Since then, they became one of the most commonly used tools in text processing.
Standard data structures answer LCE queries in constant time
and take linear space. The original construction algorithm~\cite{DBLP:journals/jcss/LandauV88,DBLP:conf/focs/Weiner73,DBLP:journals/siamcomp/HarelT84}
works in linear time for constant alphabets only,
but it has been subsequently generalized to larger integer alphabets~\cite{DBLP:journals/jacm/Farach-ColtonFM00}
and simplified substantially~\cite{DBLP:journals/jacm/KarkkainenSB06,DBLP:journals/jal/BenderFPSS05}.
Thus, LCE queries are completely resolved in the classic setting 
where the text $T$ is stored in $\Oh(n)$ space.

However, if $T$ is over a small alphabet of size $\sigma$,
then it can be stored in $\Oh(n\log \sigma)$ bits.
Yet, until very recently, even for the binary alphabet there was no
data structure of $o(n \log n)$ bits supporting LCE queries in constant time.
The first such solutions are by Tanimura et al.~\cite{DBLP:conf/mfcs/TanimuraNBIT17}
and Munro et al.~\cite{DBLP:journals/corr/abs-1712-07431}, who showed that
constant-time queries can be implemented using data structures of size
$\Oh(n\log \sigma / \sqrt{\log n} )$ and $\Oh(n\sqrt{\log  \sigma}/\sqrt{\log n})$, respectively. 
The latter result admits an $\Oh(n /\sqrt{\log_\sigma n})$-time construction from the packed
representation of~$T$.  In yet another study, Birenzwige et
al.~\cite{ElyLCE} considered LCE queries in a model
where $T$ is available for read-only random access, but not counted towards the data structure size.
Constant-time LCE queries in the optimal space of $\Oh(n \log \sigma)$ bits can be deduced as a
corollary of their results, but the construction algorithm is randomized and takes $\Oh(n)$ time.

Our contribution in the area of LCE queries is a data structure of the optimal
size $\bigO(n/\log_\sigma n)$ that answers LCE in
$\bigO(1)$ time and, furthermore, can be deterministically constructed
in the optimal $\bigO(n/\log_\sigma n)$ time. This significantly improves the
state of the art and essentially closes the LCE problem also in the packed setting.

\paragraph{Our Techniques}
Our main innovation and the key tool behind both our results is a novel notion of \emph{string synchronizing sets},
which relies on \emph{local consistency}---the idea to make symmetry-breaking decisions involving a position $i$ of the text $T$ based on the characters at the nearby positions. This way, we can guarantee that equal fragments of the text are handled in the same way.
The classic implementations of local consistency involve parsing the text; see e.g.~\cite{FirstConsistentParsing,DBLP:journals/jacm/Jez16}.
Unfortunately, the context size at a given level of the parsing is expressed in terms of the number of phrases,
whose lengths may vary significantly between regions of the text.
To overcome these limitations, Kociumaka et al.~\cite{DBLP:conf/soda/KociumakaRRW15} introduced \emph{samples assignments}
with fixed context size. Birenzwige et al.~\cite{ElyLCE} then applied the underlying techniques to define \emph{partitioning sets},
which they used for answering LCE queries. Moreover, they obtained an alternative construction of partitioning sets (with slightly inferior properties) by carefully modifying the parsing scheme of \cite{FirstConsistentParsing}. 
In his PhD thesis~\cite{PhD}, the second author introduced \emph{synchronizing functions}, an improved version of  \emph{samples assignments}
with stronger properties and efficient deterministic construction procedures.
He also used synchronizing functions to develop the optimal LCE data structure in a packed text. 
In this work, we reproduce the latter result using \emph{synchronizing sets},
which are closely related to synchronizing functions, but enjoy a much simpler and cleaner interface.

\paragraph{Organization of the Paper}
After introducing the basic notation and tools
in~\cref{sec:prelim}, we start by defining the main concept of the
paper---the string synchronizing set---and proving some of its
properties (\cref{sec:def-s}).  Next, we show how to sort suffixes in
such a set (\cref{sec:sort-s}) and extend these ideas
into an optimal LCE data structure
(\cref{sec:lce}).  We then describe how to build the BWT given a small
string synchronizing set (\cref{sec:bwt}) and prove the conditional
optimality of our construction
(\cref{sec:conditional-optimality}). We conclude by showing efficient
algorithms for the construction of string synchronizing set
(\cref{sec:construction-of-set-S}).
\begin{camera}
Due to space limitations, proofs of the claims marked with \fullonly\ are presented only in the full version of the paper~\cite{full}. 
\end{camera}

\section{Preliminaries}
\label{sec:prelim}

Let $T\in\Sigma^{*}$ be a string over alphabet $\Sigma=[0\dotdot \sigma-1]$.
Unless ex\-pli\-ci\-tly stated otherwise, we assume $\smash{\sigma=n^{\bigO(1)}}$,
where $n=|T|$.~For  $1 \leq i \leq j \leq n$, we write $T[i\dotdot j]$ to denote the
\emph{substring} $T[i]T[i+1]\cdots T[j]$. Throughout, we
use $[i\dotdot j)$ as a shorthand for $[i\dotdot j-1]$.
The length of the longest common prefix of $X,Y\in \Sigma^{*}$ is denoted
$\lcp(X,Y)$.

An integer $p\in [1\dd |X|]$ is a \emph{period} of $X$ if $X[i]=X[i+p]$ for $i\in [1\dd |X|-p]$. The shortest
period of $X$ is denoted as $\per(X)$.

\begin{lemma}[Periodicity Lemma~\cite{fine1965uniqueness}]\label{lem:per}
  If a string $X$ has periods $p,q$ such that $p+q-\gcd(p,q)\le
  |X|$, then $\gcd(p,q)$ is also its period.
\end{lemma}

\subsection{Suffix Array and BWT}
\label{sec:sa}

The \emph{suffix array}~\cite{DBLP:journals/siamcomp/ManberM93}  $\SA[1\dotdot n]$ of a
text $T$ is a permutation
defining the lexicographic order on suffixes: $T[\SA[i]\dotdot n]\allowbreak
\prec T[\SA[j]\dotdot n]$ if $i < j$. It
takes $\Oh(n)$ space and can be constructed in $\Oh(n)$ time
\cite{DBLP:journals/jacm/KarkkainenSB06}.

Given positions $i,j$ in $T$, the \emph{Longest Common Extension} query $\LCE(i,j)$ asks for
$\lcp(T[i\dotdot n],T[j\dotdot n])$.  The standard solution
consists of the suffix array $\SA$, the inverse permutation $\smash{\SA^{-1}}$ (defined so
that $\SA[\SA^{-1}[i]]=i$), the LCP table $\LCP[2\dd n]$ (whose entries are
$\LCP[i]=\LCE(\SA[{i-1}],\SA[i])$), and a data
structure for range minimum queries built on top of the $\LCP$ table;
see~\cite{DBLP:journals/jacm/KarkkainenSB06,DBLP:journals/jacm/Farach-ColtonFM00,DBLP:journals/jal/BenderFPSS05,DBLP:journals/siamcomp/HarelT84}.

\begin{proposition}\label{prop:lce}
  LCE queries in a text $T\in [0\dd \sigma)^n$ with $\sigma = n^{\Oh(1)}$ 
  can be answered in  $\Oh(1)$  time after $\Oh(n)$-time preprocessing.
\end{proposition}

The \emph{Burrows--Wheeler transform} (BWT)~\cite{BWT} of $T[1\dotdot n]$ is
defined as $\BWT[i] = T[\SA[i] - 1]$ if $\SA[i] > 1$ and $\BWT[i] = T[n]$ otherwise.
To ensure the correct handling of boundary cases, it is often assumed that
$\BWT[\SA^{-1}[1]]$ contains a sentinel $\$\notin \Sigma$.
In this paper, we avoid this to make sure that $\BWT[1\dd n]\in [0\dd \sigma)^n$.
Our construction also returns $\SA^{-1}[1]$, though, so that the corresponding value can be set as needed.

\subsection{Word RAM Model}
\label{sec:ram}

Throughout the paper, we use the standard word RAM model of
computation~\cite{DBLP:conf/stacs/Hagerup98} with $w$-bit \emph{machine words},
  where $w \ge \log n$.

In the word RAM model, strings are typically represented as arrays, with
each character occupying a single memory cell.  Nevertheless,
a single character can be represented using $\ceil{\log \sigma}$ bits, which might be much less
than~$w$.  Consequently, one may store a text $T\in [0\dd \sigma)^n$ in
$\Oh\big(\big\lceil{\frac{n\log \sigma}{w}}\big\rceil\big)$ consecutive
memory cells.  In the \emph{packed representation} of $T$, we assume
that the first character corresponds to the $\ceil{\log\sigma}$ least
significant bits of the first cell.

\begin{camera}
\begin{proposition}[\fullonly]\label{prop:packed}
  Suppose that $T\in [0\dd \sigma)^n$ is stored in the
  packed representation. The packed representation of any
  length-$\ell$ substring can be retrieved in
  $\Oh\big(\big\lceil{\frac{\ell\log\sigma}{w}}\big\rceil\big)$ time.  The
  longest common prefix of two length-$\ell$ fragments
  can be identified in the same time.
\end{proposition}
\end{camera}
\begin{full}
  \begin{proposition}\label{prop:packed}
    Suppose that $T\in [0\dd \sigma)^n$ is stored in the
    packed representation. The packed representation of any
    length-$\ell$ substring can be retrieved in
    $\Oh\big(\big\lceil{\frac{\ell\log\sigma}{w}}\big\rceil\big)$ time.  The
    longest common prefix of two length-$\ell$ fragments
    can be identified in the same time.
  \end{proposition}
\begin{proof}
  The bit sequence corresponding to any fragment of length $\ell$ is
  contained in the concatenation of at most
  $1+\big\lceil{\frac{\ell\ceil{\log \sigma}}{w}}\big\rceil$ memory
  cells of the packed representation of $T$. Its location can be
  determined in $\Oh(1)$ time, and the resulting sequence can be
  aligned using $\Oh(\lceil{\frac{\ell\log\sigma}{w}}\rceil)$ bit-wise
  shift operations, as well as $\Oh(1)$ bit-wise and operations to
  mask out the adjacent characters.  This results in a packed
  representation of the length-$\ell$ fragment of $T$.  In order to
  compute the length of the longest common prefix of two such
  fragments, we xor the packed representations and find the position
  $p$ of the least significant bit in the resulting sequence.  The resulting length is
  $\big\lfloor{\frac{p-1}{\ceil{\log\sigma}}}\big\rfloor$ assuming
  $1$-based indexing of positions.
\end{proof}
\end{full}

A particularly important case is that of $\sigma=2$.
In many applications, these \emph{bitvectors} are equipped with a data structure
answering \emph{rank} queries: for $B[1\dotdot n]$,
  $\rank_{B}(i)= |\{j\in [1\dd i] : B[j]=1\}|$.
Jacobson~\cite{DBLP:conf/focs/Jacobson89}
proved that $\rank_B$ queries can be answered in $\Oh(1)$ time
using an additional component of $o(n)$ extra bits.
However, an efficient construction of such a component is much more
recent.

\begin{proposition}[\cite{WaveletSuffixTree,DBLP:journals/tcs/MunroNV16}]\label{prp:rksel}
  A packed bitvector $B[1\dotdot n]$ can be extended in
  $\Oh\big(\frac{n}{\log n}\big)$ time with a data structure of size $o\big(\frac{n}{\log n}\big)$
  which answers $\rank_B$ queries in $\Oh(1)$ time.
\end{proposition}

\subsection{Wavelet Trees}
\label{sec:wavelet-trees}

Wavelet trees, invented by Grossi, Gupta, and Vitter
\cite{DBLP:conf/soda/GrossiGV03} for space-efficient text indexing,
are important data structures with a vast number of applications far
beyond text processing (see \cite{DBLP:journals/jda/Navarro14}).

The wavelet tree of a string $W\in[0\dd 2^b)^n$ is recursively defined as follows. First, we create
the root node $v_{\eps}$. This completes
the construction for $b=0$. If $b>0$, we attach to $v_{\eps}$ a bitvector $B_\eps[1\dd n]$ in which
$B_\eps[i]$ is the most significant bit of $W[i]$ (interpreted as a $b$-bit
number). Next, we partition $W$ into subsequences $W_{0}$ and $W_{1}$ by scanning $W$ and appending
$W[i]$, with the most significant bit $d$ removed, to the subsequence $W_d$. Finally, we attach the recursively
created wavelet trees of $W_0$ and $W_1$ (over alphabet $[0\dd 2^{b-1})$) to
$v_{\eps}$. The result is a perfect binary tree with $2^b$ leaves.

Assuming that we label edges $0$ (resp. $1$) if they go to the left (resp. right) child, we define the
\emph{label} of a node to be the concatenation of the labels on the root-to-node path.
If $B_X$ denotes the bitvector of a node
$v_X$ labeled $X\in\{0,1\}^{<b}$, then $B_X$ contains one bit (following $X$ as a prefix)
from each $W[i]$ whose binary encoding has prefix $X$.
Importantly, the bits in the bitvector $B_X$ occur in the same order as the corresponding elements $W[i]$ occur in $W$.

It is easy to see that the space occupied by the bitvectors is $\Oh(n b)$ bits, i.e., $\Oh\big(\frac{nb}{\log n}\big)$ words.
We need one extra machine word per node for pointers and due to word alignment,
which sums up to $\Oh(2^b)$.  Thus, the total size of a wavelet tree
is $\Oh\big(2^b + \frac{n b}{\log n}\big)$ machine words, which
is $\Oh\big(\frac{n b}{\log n}\big)$ if $b \le \log n$.
As shown recently, a wavelet tree
can be constructed efficiently from the packed representation of $W$.
 
\begin{theorem}[\cite{WaveletSuffixTree,DBLP:journals/tcs/MunroNV16}]\label{thm:wavelet_tree}
  Given the packed representation of a string $W$ of length $n$ over
  $[0\dd 2^b)$ for $b \le \log n$, we can construct its
  wavelet tree in $\Oh(n b / \sqrt{\log n})$ time using $\Oh(n b / \log n)$ space.
\end{theorem}

\section{String Synchronizing Sets}
\label{sec:def-s}

In this section, we introduce string synchronizing sets, the central novel concept
 underlying both main results of this paper.
\begin{definition}
\label{def:sss}
  Let $T$ be a string of length $n$ and let $\tau\le \frac12n$ be a positive integer.
  We say that a set $\S \subseteq [1\dd n-2\tau+1]$ is a
  \emph{$\tau$-synchronizing set} of $T$ if it satisfies the following conditions:

  \begin{enumerate}
  \item\label{it:cons} if $T[i\dotdot i+2\tau)=T[j\dotdot j+2\tau)$, then $i \in \S$ holds if
    and only if $j \in \S$ (for $i,j\in[1\dd n-2\tau+1]$), and
    \item\label{it:dom} $\S\cap[i\dd i+\tau)=\emptyset$ if and only if
    $i\in \R$ (for $i\in[1\dd n-3\tau+2]$), where
    \[
  \R = \{i\in[1 \dd n-3\tau+2] : \per(T[i\dotdot i+3\tau-2]) \le \tfrac13\tau\}.
\]
  \end{enumerate}
\end{definition}

Intuitively, the above definition requires that the decision on
whether $i\in \S$ depends entirely on $T[i\dotdot i+2\tau)$, i.e., it is made
consistently across the whole text (the first \emph{consistency
  condition}) and that $\S$ contains densely
  distributed positions within (and only within) non-periodic regions of
  $T$ (the second \emph{density condition}).

The properties of a $\tau$-synchronizing set $\S$ allow for
symmetry-breaking decisions that let us individually process only positions
$i\in \S$, compared to the classic $\Oh(n)$-time algorithms handling
all positions one by one.  Thus, we are interested in minimizing the
size of $\S$.  Since $\R=\emptyset$ is possible in general, the smallest
$\tau$-synchronizing set we can hope for is of size $\Omega(\frac{n}{\tau})$ in the worst case.
Our deterministic construction in \cref{sec:construction-of-set-S} matches this
lower bound.

Note that the notion of a $\tau$-synchronizing set is valid for every positive integer
$\tau \le \frac12 n$. Some applications make use of~many
\mbox{synchronizing} sets
with parameters $\tau$ spread across the whole domain; see~\cite{CPM2019,PhD}. 
However, in this paper we only rely on $\tau$-synchronizing sets for
$\tau=\eps\log_{\sigma}n$ (where $\eps$ is a sufficiently small positive
constant), because this value turns out to be the suitable for processing the
packed representation of a text $T\in[0 \dd \sigma)^n$ stored in $\Theta(n/\log_{\sigma} n)$ machine words. 
This is because our generic construction algorithm (\cref{lem:sync}) runs in $\Oh(n)$ time,
whereas a version optimized for packed strings (\cref{lem:sync_fast}) takes $\Oh(\frac{n}{\tau})$ time
only for $\tau \le \eps \log_\sigma n$ with $\eps < \frac15$. 
(Note that an $\Oh(\frac{n}{\tau})$-time construction is feasible for $\tau = \Oh(\log_\sigma n)$ only,
  because we need to spend $\Omega(n/\log_{\sigma}n)$ time already to read the whole packed text.)
Moreover, the running time of our BWT construction procedure involves a term $\sigma^{\Oh(\tau)}$,
which would dominate if we set $\tau$ too large.

We conclude this section with two properties of $\tau$-synchronizing sets
useful across all our applications.
To formulate them, we define the \emph{successor} in $\S$ for each $i\in [1\dd n-2\tau+1]$:
\[ \suc_\S(i) := \min\{j \in \S \cup \{n-2\tau+2\} : j \ge i\}.\]
The sentinel $n-2\tau+2$ guarantees that the set on the right-hand side is non-empty.
Our first result applies the density condition to relate $\suc_\S(i)$ for $i\in \R$ with maximal periodic regions of $T$.
\begin{fact}\label{fct:den}
  Let $T$ be a text and let $\S$ be its $\tau$-synchronizing set for a positive integer $\tau\le \frac12$. 
  If $i\in \R$ and $p = \per(T[i\dd i+3\tau-2])$, then $T[i\dd \suc_\S(i)+2\tau-2]$
  is the longest prefix of $T[i\dd n]$ with period $p$.
\end{fact}
\begin{proof}
  Let us define $s = \suc_\S(i)$ and observe that $[i\dd s)\cap \S = \emptyset$.
  Consequently, $[j\dd j+\tau)\cap \S = \emptyset$ holds for every $j\in [i\dd s-\tau]$.
  By the density condition, this implies $[i\dd s-\tau]\sub \R$,
  i.e., that $\per(T[j\dd j+3\tau-2])\le \tfrac13\tau$ for $j\in [i\dd s-\tau]$.
  We shall prove by induction on $j$ that  $p=\per(T[i\dd j+3\tau-2])$ holds for $j\in [i\dd s-\tau]$. 
  Tha base case of $j=i$ follows from the definition of $p$. For $j > i$, on the other hand,
  let us denote $p' = \per(T[j\dd j+3\tau-2])$ and assume $p = \per(T[i\dd j+3\tau-3])$ by the inductive hypothesis.
  We observe that $j+3\tau-2-p'-p > j+2\tau-2 \ge j$,
  so $T[j+3\tau-2] = T[j+3\tau-2-p'] = T[j+3\tau-2-p'-p] = T[j+3\tau-2-p]$.
  This shows that $p=\per(T[i\dd j+3\tau-2])$ and completes the inductive step.
  We conclude that $p$ is the shortest period of $T[i\dd s+2\tau-2]$.
  We now need to prove that this is the longest prefix of $T[i\dd n]$ with period $p$.
  The claim is trivially true if $s = n-2\tau+2$.
  Otherwise, $s\in \S$, so $[s-\tau+1\dd s]\cap \S\ne \emptyset$.
  By the density condition, this means that $s-\tau+1\notin \R$,
  i.e., $\per(T[s-\tau+1\dd s+2\tau-1])>\frac13\tau$.
  As a result, $\per(T[i\dd s+2\tau-1])>\frac13\tau \ge p$. This completes the proof.
\end{proof}

The second result applies the consistency condition to relate $\suc_\S(i)$ with $\suc_\S(j)$
for two common starting positions $i,j\in [1\dd n-2\tau+1]$ of a sufficiently long substring. 
\begin{fact}\label{fct:cons}
  Let $T$ be a text and let $\S$ be its $\tau$-synchronizing set for a positive integer $\tau\le \frac12$. 
  If a substring $X$ of length $|X|\ge 2\tau$ occurs in $T$ at positions $i$ and $j$,
  then either
  \begin{enumerate}[label={(\roman*)}]
    \item\label{alt:one} $\suc_\S(i)-i = \suc_\S(j)-j \le |X|-2\tau$, or
    \item\label{alt:two} $\suc_\S(i)-i > |X|-2\tau$ and $\suc_\S(j)-j > |X|-2\tau$.
  \end{enumerate}
  Moreover, \ref{alt:one} holds if $|X|\ge 3\tau-1$ and $\per(X)> \frac13\tau$.
\end{fact}
\begin{proof}
First, we shall prove that \ref{alt:one} holds if \ref{alt:two} does not.
Without loss of generality, we assume that $\suc_\S(i)-i \le \suc_\S(j)-j$,
which yields $\suc_\S(i)-i \le |X|-2\tau$. In particular, $T[i\dd \suc_\S(i)+2\tau)=T[j\dd j-i+\suc_\S(i)+2\tau)$ is a prefix of $X$.
Moreover, $\suc_\S(i) \le n-2\tau+1$, so $\suc_\S(i)\in \S$.
The consistency condition therefore implies $j-i+\suc_\S(i)\in \S$.
Hence, $\suc_\S(j)\le j-i+\suc_\S(i)$, and $\suc_\S(i)-i = \suc_\S(j)-j$ thus holds as claimed.

Next, we shall prove that $\suc_\S(i)-i \le |X|-2\tau$ if $|X|\ge 3\tau-1$ and $\per(X)> \frac13\tau$.
From this, we shall conclude that \ref{alt:two} does not hold (whereas \ref{alt:one} holds) in that case.
If $i\notin \R$, then $[i\dd i+\tau)\cap \S \ne \emptyset$ by the density condition, so $\suc_\S(i)-i \le \tau-1 \le |X|-2\tau$.
Otherwise, let us define $p=\per(T[i\dd i+3\tau-2])$ and note that $T[i\dd \suc_\S(i)+2\tau-2]$ has period $p$ by \cref{fct:den}.
Since $p\le \frac13 n$, this means that $|X| > |T[i\dd \suc_\S(i)+2\tau-2]|$,
which is equivalent to the desired inequality $\suc_\S(i)-i \le |X|-2\tau$.
\end{proof}

\section{Sorting Suffixes Starting in Synchronizing Sets}
\label{sec:sort-s}

Let $T\in [0\dd \sigma)^n$ be a text stored in the packed representation
and let $\S$ be its $\tau$-synchronizing set of size $\Oh(\frac{n}{\tau})$
for $\tau = \Oh(\log_{\sigma} n)$. 
In this section, we show that given the above as input,
the suffixes of $T$ starting at positions in $\S$
can be sorted lexicographically in the optimal $\Oh(\frac{n}{\tau})$ time.
We assume that the elements of $\S$ are stored in an array in the left-to-right
order so that we can access the $i$th smallest element, denoted $s_i$,
in constant time for $i\in[1\dd |\S|]$.  The presented algorithm is the
first step in our BWT construction. It also reveals the key ideas behind our LCE data structure.

\subsection{The Nonperiodic Case}
\label{sec:sort-s-simple}

Consider first a case when $\R=\emptyset$.  The density condition
then simplifies to the following statement:
\begin{camera}
\begin{enumerate}[start=2,label=$({\arabic*}')$]
  \item $\S \cap [i\dd i+\tau) \neq \emptyset$ for every
    $i\in[1\dd n-3\tau+2]$.
\end{enumerate}
\end{camera}
\begin{full}
  \begin{enumerate}[start=2,label=${\arabic*}'.$]
    \item $\S \cap [i\dd i+\tau) \neq \emptyset$ for every
      $i\in[1\dd n-3\tau+2]$.
  \end{enumerate}
\end{full}
We introduce a string $T'$ of length $n':=|\S|$ defining it so that $T'[i] = T[s_i\dd \min(n, s_i+3\tau-1)]$.
All characters of $T'$ are strings over $[0\dd \sigma)$ of length up to $3\tau$. Hence, they can be encoded using $\Oh(\tau\log \sigma)=\Oh(\log n)$-bit
integers so that the lexicographic order is preserved.\footnote{For example, we may append $6\tau-2|T'[i]|$ zeroes
and $|T'[i]|$ ones to $T'[i]$. The result can then be interpreted as the base-$\sigma$ representation of an integer in $[0\dd \sigma^{6\tau})$.}
Furthermore, the lexicographic order of the suffixes of $T'$ coincides with that
of the corresponding suffixes of~$T$.

\begin{lemma}\label{lm:sort-s-simple}
  Assume $\R=\emptyset$ holds for a text $T$.
  If positions $i,j$ of $T'$ satisfy $T'[i\dd n'] \prec T'[j\dd n']$, then $T[s_i \dd n] \prec T[s_j \dd n]$.
\end{lemma}
\begin{proof}
  We proceed by induction on $\lcp(T'[i\dotdot n'],T'[j\dotdot n'])$.  
  The base case is that $T'[i]\prec T'[j]$. If $T'[i]$ is a proper prefix of $T'[j]$, then $T[s_{i}\dd n]=T'[i]\prec T'[j] \preceq T[s_j\dd n]$.
  Otherwise, $T'[i]\cdot W \prec T'[j]$ holds for any string $W$,
  so $T[s_i\dd n] \prec T'[j] \preceq T[s_j \dd n]$. 

  Henceforth, we assume that $T'[i]=T'[j]$. Since $i\ne j$, this implies $|T'[i]|=|T'[j]|=3\tau$
  and $s_i, s_j \le  n-3\tau+1$.   The density condition
  yields $s_{n'}\ge n-3\tau+2$  (due to $n-3\tau+2\notin \R$), so we further have $i,j\in [1\dd n')$ and $T'[i+1\dd n'] \prec T'[j+1\dd n']$.
  Moreover, 
   $X:= T[s_i+1\dd s_i+3\tau)=T[s_j+1\dd s_j+3\tau)$ occurs in $T$ at positions $s_i+1$ and $s_j+1$.
  As $\per(X)>\frac13\tau$ (due to $s_{i}+1 \notin \R$), \cref{fct:cons} implies
  \[\suc_\S(s_i+1)-(s_i+1) = \suc_\S(s_j+1)-(s_j+1) \le \tau-1,\]
  i.e., $s_{i+1}-s_i = s_{j+1}-s_j \le \tau$. Furthermore, $T[s_i\dd s_{i+1})=T[s_j \dd s_{j+1})$
  because  $T[s_i\dd {s_i+3\tau})=T'[i]=T'[j]=T[s_j\dd {s_j+3\tau})$.
  Due to $T[s_{i+1}\dd n] \prec T[s_{j+1}\dd n]$ (which we derive from the inductive hypothesis),
  this implies $T[s_i \dd n] \prec T[s_j\dd n]$ and completes the proof of the inductive step.
\end{proof}

By \cref{lm:sort-s-simple}, the suffix array of $T'$ can be used
to retrieve the lexicographic order of the suffixes $T[s_i\dotdot n]$ for $s_i\in \S$.
Recall that each symbol of $T'$ takes $\Oh(\tau\log \sigma)=\Oh(\log n)$
bits, so the suffix array of $T'$ can be computed in $\bigO(|T'|)=\bigO(\frac{n}{\tau})$
time~\cite{DBLP:journals/jacm/KarkkainenSB06}.

\subsection{The General Case}
\label{sec:sort-s-general}

We now show how to adapt the approach from the previous section so that it also 
works if $\R\ne\emptyset$. As before, we
construct a string $T'$ of length $n'=|\S|$ over a polynomially bounded integer alphabet, and we sort its suffixes.
However, the definition of $T'$ becomes more involved.
To streamline the formulae, we set $s_{n'+1}=n-2\tau+2$.
For each $i\in [1\dd n']$, we define $T'[i]=(T[s_i\dotdot \min(n,s_i+3\tau-1)], d_i)$, where $d_i$ is an integer specified as follows:

\begin{enumerate}[label=(\alph*)]
\item If $s_{i+1}-s_{i}\leq\tau$ (in particular, if $s_i> n-3\tau+1$), then $d_i = 0$.
\item Otherwise, we set $p_i=\per(T[s_i+1\dotdot s_i+3\tau))$
and
  \[ d_i=
    \begin{cases}
      n-s_{i+1}+s_i &\!\text{if } T[s_{i+1}+2\tau-1]\succ T[s_{i+1}+2\tau-1-p_i],\\
      s_{i+1}-s_i-n    &\!\text{otherwise (if $s_{i+1}=n-2\tau+2$ in particular).}
    \end{cases}
  \]
\end{enumerate}
Note that each $T'[i]$ can be encoded in $\Oh(\tau\log\sigma + \log n) = \bigO(\log n)$ bits so that the comparison of the resulting integers is
equivalent to the lexicographic comparison of the corresponding
symbols. 

\begin{lemma}\label{lm:sort-s-general}
  If positions $i,j$ of $T'$ satisfy $T'[i\dd n'] \prec T'[j\dd n']$, then $T[s_i \dd n] \prec T[s_j \dd n]$.
\end{lemma}
\begin{proof}
  Induction on   $\lcp(T'[i\dotdot n'],  T'[j\dotdot n'])$. 
  If $\lcp(T[s_i\dd n],\allowbreak T[s_j\dd n]) < 3\tau$,
  then we proceed as in the proof of \cref{lm:sort-s-simple}.

  Otherwise, the string $X=T[s_i+1\dd s_i+3\tau)=T[s_j+1\dd s_j+3\tau)$ occurs in $T$ at positions $s_i+1$ and $s_j+1$.
  If $\min(s_{i+1}-s_i,s_{j+1}-s_j) \le \tau$, 
  then \cref{fct:cons} yields $s_{i+1}-s_i = s_{j+1}-s_j \le  \tau$,
  so $d_i = d_j = 0$ and $T'[i]=T'[j]$.
  Moreover, $i,j \in [1\dd n')$ due to $s_{i},s_j \le n-3\tau+1$.
  Consequently, the claim follows from $T[s_i\dd s_{i+1})=T[s_j\dd s_{j+1})$ because the inductive hypothesis
  yields $T[s_{i+1}\dd n]\prec T[s_{j+1}\dd n]$.

  On the other hand, $\min(s_{i+1}-s_i,s_{j+1}-s_j) > \tau$
  yields $d_i, d_j \ne 0$. Moreover, the density condition implies $s_i+1, s_j+1\in \R$
  with $p_i = p_j = \per(X)\le \frac13\tau$.
    By \cref{fct:den}, the longest prefix of $T[s_i+1\dd n]$ with period $p_i$ is
    $P_i := T[s_i+1\dd s_{i+1}+2\tau-2]$ and the longest prefix of $T[s_j+1\dd n]$ with period $p_j$
    is $P_j := T[s_j+1\dd s_{j+1}+2\tau-2]$. Both $P_i$ and $P_j$ start with $X$, so
    one of them is a prefix of the other.
    We consider three cases based on how their lengths, $|P_i|=n+2\tau-2-|d_i|$ and $|P_j|=n+2\tau-2-|d_j|$, 
    compare to each other.
    \begin{itemize}
      \item If $|d_i|>|d_j|$, then $P_i$ is a proper prefix of $P_j$.
      If $i=n'$, then $T[s_i\dd n] = P_i \prec P_j \preceq T[s_j\dd n]$.
      Otherwise, we note that $d_i < 0$ due to $d_i < d_j$, so $T[s_i+1+|P_i|] \prec \allowbreak P_i[|P_i|-p_i+1]=P_j[|P_i|-p_j+1]=
      T[s_j+1+|P_i|]$, which yields the claim.
        \item If $|d_i|=|d_j|$, then $P_i = P_j$.
        If $i=n'$, then $T[s_i\dd n] = P_i = P_j \prec T[s_j\dd n]$. 
        Otherwise, we consider two subcases:
        \begin{itemize}
        \item
         If $d_i=-d_j$, then $d_i  < 0 < d_j$, so $T[s_{i}+1+|P_i|] \prec P_i[|P_i|-p_i+1]=P_j[|P_j|-p_j+1] \prec 
        T[s_{j}+1+|P_j|]$, which also yields the claim.
        \item 
        Finally, if $d_i = d_j$, then $T'[i]=T'[j]$ and $i,j\in [1\dd n')$,
        so the inductive hypothesis gives $T[s_{i+1}\dd n]\prec T[s_{j+1}\dd n]$. 
        The claim follows due to $T[s_i\dd s_{i+1})=T[s_j\dd s_{j+1})$.
        \end{itemize}
      \item If $|d_i|<|d_j|$, then $P_j$ is a proper prefix of $P_i$.
      Moreover, $d_j > 0$ due to $d_i < d_j$, so $T[s_i+1+|P_j|] = P_i[|P_j|-p_i+1] = P_j[|P_j|-p_j+1]\prec T[s_{j}+1+|P_j|]$
      and the claim holds.\qedhere
    \end{itemize}
\end{proof}

We now prove that efficient construction of $T'$ is indeed
possible. The only difficulty is computing the values $p_i$ in case $s_{i+1}-s_i > \tau$.
To achieve this in constant time, we observe that $p_i \le \frac13\tau$ holds by the density condition
due to $s_i < n-3\tau+2$ and $s_{i}+1\notin \R$.
Consequently, $p_i$ is also the shortest period of every prefix of $T[s_i+1\dd s_{i}+3\tau)$
of length $2p_i$ or more.
By the synchronizing property of primitive strings~\cite[Lemma 1.11]{AlgorithmsOnStrings},
this means that the leftmost occurrence of $T[s_i+1 \dd s_i+\tau]$ in $T[s_i+2 \dd s_i+2\tau]$
starts at position $p_i$.
We can find it in $\Oh(1)$ time (after $\Oh(n^\eps)$-time preprocessing)
using the packed string matching algorithm~\cite{DBLP:journals/tcs/Ben-KikiBBGGW14}.

\begin{theorem}
  \label{thm:sort-s}
  Given the packed representation of a text $T\in[0\dd \sigma)^n$ and its $\tau$-synchronizing set $\S$ of
  size $\bigO(\frac{n}{\tau})$ for $\tau =\Oh(\log_{\sigma} n$), we can compute in $\bigO(\frac{n}{\tau})$ time the
  lexicographic order of all suffixes of $T$ starting at positions
  in $\S$.
\end{theorem}

\section{Data Structure for LCE Queries}
\label{sec:lce}

In \cref{sec:sort-s}, for a text $T\in [0\dd \sigma)^n$ and its $\tau$-synchronizing set $\S$
with $\tau=\Oh(\log_\sigma n)$, we constructed a string $T'$ such the lexicographic order 
of the suffixes of $T'$ coincides with the order of suffixes of $T$ starting at positions in $\S$.
In this section, we show how to reduce LCE queries in $T$ to LCE queries in $T'$.
Our approach results in a data structure with $\Oh(1)$-time LCE queries
and $\Oh(\frac{n}{\tau})$-time construction provided that $|\S|=\Oh(\frac{n}{\tau})$.
Recall that $n'=|\S|=|T'|$, $s_i$ is the $i$th smallest element of $\S$,
and $s_{n'+1}=n-2\tau+2$.

\subsection{The Nonperiodic Case}
\label{sec:lce-simple}

Analogously to \cref{sec:sort-s-simple}, we start with the case of $\R=\emptyset$,
which makes the definition of $T'$ simpler: $T'[i]=T[s_i\dd \min(n,s_i+3\tau-1)]$.

Consider an LCE query in the text $T$. 
If $\LCE(i,j) < 3\tau$, we can retrieve it in $\Oh(1)$
time from the packed representation of $T$.
Otherwise, \cref{fct:cons} yields $\suc_\S(i)-i = \suc_\S(j)-j < \tau$.
Hence, $\LCE(i,j)=s_{i'}-i+\LCE(s_{i'}, s_{j'})$,  where $s_{i'} = \suc_\S(i)$ and $s_{j'} = \suc_\S(j)$.
A similar reasoning can be repeated to determine $\LCE(s_{i'}, s_{j'})$,
which must be smaller than $3\tau$ or equal to $s_{i'+1}-s_{i'}+\LCE(s_{i'+1}, s_{j'+1})$.
The former condition can be verified by checking whether $T'[i']=T'[j']$. 
A formal recursive application of this argument results in the following characterization:

\begin{fact}\label{fct:lce_simple}
  Consider a string $T\in [0\dd \sigma)^n$ which satisfies $\R=\emptyset$.
  For positions $i,j$ in $T$ such that $\LCE(i,j)\ge 3\tau-1$,
  let us define $s_{i'}=\suc_\S(i)$ as well as $s_{j'}=\suc_\S(j)$. 
  If $\ell = \LCE_{T'}(i',j')$, then
   \[\LCE(i,j)=  s_{i'+\ell}-i+\LCE(s_{i'+\ell},s_{j'+\ell})< s_{i'+\ell}-i+3\tau.\]
\end{fact}
\begin{proof}
  The proof is by induction on $\ell$.
  Due to $i,j\notin \R$, \cref{fct:cons} yields $s_{i'}-i=s_{j'}-j < \tau$,
  and therefore $T[i\dd s_{i'})=T[j \dd s_{j'})$. 
  Hence, $\LCE(i,j)=  s_{i'}-i+\LCE(s_{i'},s_{j'})$.
  If $\ell = 0$, it just remains to prove that $\LCE(s_{i'},s_{j'})<3\tau$,
  which follows from $T'[i']\ne T'[j']$.

  For $\ell > 0$, we note that $T'[i']=T'[j']$, so $\LCE(s_{i'},s_{j'})\ge 3\tau$
  and $\LCE(s_{i'}+1, s_{j'}+1)\ge 3\tau-1$.
  The inductive hypothesis now yields
  $\LCE(s_{i'}+1, s_{j'}+1)=s_{i'+\ell}-s_{i'}-1 + \LCE(s_{i'+\ell},s_{j'+\ell})$
  and $\LCE(s_{i'+\ell},s_{j'+\ell}) < 3\tau$.
  Since $\LCE(i,j)=  s_{i'}-i + \LCE(s_{i'},s_{j'})=s_{i'}+1-i+\LCE(s_{i'}+1,s_{j'}+1)$, this completes the proof.
\end{proof}

\cref{fct:lce_simple} leads to a data structure for LCE queries that consists of
the packed representation of $T$ (\cref{prop:packed}), a
$\tau$-synchronizing set $\S$ of size $\Oh(\frac{n}{\tau})$, a component for LCE queries in $T'$
(\cref{prop:lce}; the alphabet size is
$\sigma^{3\tau}=n^{\bigO(1)}$), and a bitvector $B[1\dotdot n]$, with
$B[i]=1$ if and only if $i\in \S$, equipped with a component for
$\Oh(1)$-time rank queries (\cref{prp:rksel}). 

To compute $\LCE(i,j)$, we first use the packed
representation to retrieve the answer in $\bigO(1)$ time provided that
$\LCE(i,j)\,{<}\,3\tau$. Otherwise,
we obtain $i'$ and $j'$ such that
$s_{i'}=\suc_{\S}(i)$ and $s_{j'}=\suc_{\S}(j)$ using $\rank_B$ queries, and we compute
$\ell=\LCE_{T'}(i',j')$. By \cref{fct:lce_simple},
$\LCE(i,j)=s_{i'+\ell}-i+ \LCE(s_{i'+\ell},s_{j'+\ell})$.  Since
$\LCE(s_{i'+\ell},s_{j'+\ell})<3\tau$,
we finalize the algorithm using the packed representation again.

\subsection{The General Case}
\label{sec:lce-general}

In this section, we generalize the results of \cref{sec:lce-simple} so that the case of $\R\ne \emptyset$ is also handled.
Our data structure consists of the same components; the only difference is that the string $T'$
is now defined as in \cref{sec:sort-s-general} rather than as in \cref{sec:sort-s-simple}.

The query algorithm needs more changes but shares the original outline.
If $\LCE(i,j)<3\tau$, then we determine the answer using \cref{prop:packed}.
Otherwise, we apply the following lemma as a reduction to computing $\LCE(\suc_\S(i),\suc_\S(j))$.
\begin{camera}
\begin{lemma}[\fullonly]\label{lem:lce}
  For positions $i,j$ in $T$ such that $\LCE(i,j)\ge 3\tau-1$,
  let us define $s_{i'}=\suc_\S(i)$ and $s_{j'}=\suc_\S(j)$.
  Then
    \[\LCE(i,j) = \begin{cases}
      \min(s_{i'}-i, s_{j'}-j)+2\tau-1 &\text{ if }s_{i'}-i \ne s_{j'}-j,\\
      s_{i'}-i + \LCE(s_{i'}, s_{j'}) & \text{ if }s_{i'}-i = s_{j'}-j.
    \end{cases}
      \]
  \end{lemma}
\end{camera}
\begin{full}
  \begin{lemma}\label{lem:lce}
    For positions $i,j$ in $T$ such that $\LCE(i,j)\ge 3\tau-1$,
    let us define $s_{i'}=\suc_\S(i)$ and $s_{j'}=\suc_\S(j)$.
    Then
      \[\LCE(i,j) = \begin{cases}
        \min(s_{i'}-i, s_{j'}-j)+2\tau-1 &\text{ if }s_{i'}-i \ne s_{j'}-j,\\
        s_{i'}-i + \LCE(s_{i'}, s_{j'}) & \text{ if }s_{i'}-i = s_{j'}-j.
      \end{cases}
        \]
    \end{lemma}
\begin{proof}
  If $\min(s_{i'}-i,s_{j'}-j)< \tau$, then $\min({s_{i'}-i},s_{j'}-j) \le \LCE(i,j) - 2\tau$,
  so  \cref{fct:cons} yields $s_{i'}-i=s_{j'}-j < \tau$.
  Moreover, $T[i\dd s_{i'}+2\tau)=T[j\dd s_{j'}+2\tau)$ and, in particular,
  $T[i\dd s_{i'})=T[j\dd s_{j'})$. Hence, $\LCE(i,j) = s_{i'}-i +\LCE(s_{i'}, s_{j'})$ holds as claimed.

  We now assume that $\min(s_{i'}-i, s_{j'}-j) \ge \tau$.
  Then $i,j\in \R$,
  and $T[i\dd i+3\tau-2] = T[j\dd j+3\tau-2]$ have the same shortest period $p\le \frac13\tau$.
  By \cref{fct:den}, the longest prefix of $T[i\dd n]$ with period $p$ is $T[i\dd s_{i'}+2\tau-2]$
  and the longest prefix of $T[j\dd n]$ with period $p$ is $T[j\dd s_{j'}+2\tau-2]$.
  In particular, one of these prefixes is a prefix of the other.
  If $s_{i'}-i \ne s_{j'}-j$, then the longest common prefix of $T[i\dd n]$ and $T[j\dd n]$
  is the shorter of the two prefixes with period $p$. Hence, $\LCE(i,j)=\min(s_{i'}-i, s_{j'}-j)+2\tau-1$ holds as claimed.
  Otherwise, $T[i \dd s_{i'}+2\tau-2]=T[j \dd s_{j'}+2\tau-2]$ yields $T[i\dd s_{i'})=T[j\dd s_{j'})$ and thus also the claim.
\end{proof}
\end{full}

We are left with determining the values $\LCE(s_i,s_j)$ for $i,j\in [1\dd n'+1]$,
i.e., handling LCE queries for positions in $\S\cup\{n-2\tau+2\}$.
The next result reduces this task to answering LCE queries in~$T'$.

\begin{camera}
\begin{lemma}[\fullonly]\label{lem:lce2}
  If $\ell = LCE_{T'}(i,j)$ for positions $i,j\in [1\dd {n'+1}]$,
  then $\LCE(s_{i},s_{j}) = s_{i+\ell}-s_{i}+\LCE(s_{i+\ell},s_{j+\ell})$.
  Moreover, $\LCE(s_{i},s_{j})< 3\tau$ or $\LCE(s_{i},s_{j}) = \min(s_{i+1}-s_{i}, s_{j+1}-s_{j})+2\tau-1$ holds if $\ell=0$.
\end{lemma}
\end{camera}
\begin{full}
  \begin{lemma}\label{lem:lce2}
    If $\ell = LCE_{T'}(i,j)$ for positions $i,j\in [1\dd {n'+1}]$,
    then $\LCE(s_{i},s_{j}) = s_{i+\ell}-s_{i}+\LCE(s_{i+\ell},s_{j+\ell})$.
    Moreover, $\LCE(s_{i},s_{j})< 3\tau$ or $\LCE(s_{i},s_{j}) = \min(s_{i+1}-s_{i}, s_{j+1}-s_{j})+2\tau-1$ holds if $\ell=0$.
  \end{lemma}
\begin{proof}
  We prove the first claim inductively. The base case of $\ell = 0$ holds trivially,
  so let us consider $\ell > 0$.
  We then have $i,j\in [1\dd n']$ and $T'[i]=T'[j]$.
  This equality yields $\LCE(s_{i},s_{j})\ge 3\tau$, so we may use \cref{lem:lce} for $s_i+1$ and $s_j+1$ to obtain $\LCE(s_{i},s_{j}) = s_{i+1}-s_{i}+\LCE(s_{i+1},s_{j+1})$
  provided that $s_{i+1}-s_{i}=s_{j+1}-s_{j}$.
  The latter equality follows from \cref{fct:cons} if $d_i = d_{j} = 0$ and from $d_i = d_j$ otherwise.
  We derive the final claim by applying the inductive hypothesis for $i+1$ and $j+1$;
  note that $\LCE_{T'}(i+1,j+1)=\ell-1$.

  Next, let us prove the second claim for $\ell=0$. It holds trivially if $\LCE(s_{i},s_{j})< 3\tau$.
  Otherwise, $i,j\in [1\dd n']$, which implies $T'[i]\ne T'[j]$ and  $d_i \ne d_j$.
  Since \cref{fct:cons} gives equivalence between $d_i=0$ and $d_j=0$, we conclude that $d_i,d_j \ne 0$.
  If $s_{i+1}-s_{i} \ne s_{j+1}-s_j$, then we may use \cref{lem:lce} to prove
   $\LCE(s_{i}+1,s_{j}+1) = \min(s_{i+1}-s_{i}-1, s_{j+1}-s_{j}-1)+2\tau-1$,
  which yields the claimed equality.
  Otherwise, we must have $d_i=-d_j$; assume $d_i<0<d_j$ without loss of generality.
  By \cref{fct:den}, $T[s_{i}\dd s_{i+1}+2\tau-2]=T[s_{j}\dd s_{j+1}+2\tau-2]$ has
  period $p=\per(T[s_i+1\dd s_i+3\tau))$.
  Furthermore, $i=n'$ and hence $s_{i+1}+2\tau-2=n$, or
  $T[s_{i+1}+2\tau-1]\prec T[s_{i+1}+2\tau-1-p]=T[s_{j+1}+2\tau-1-p]\prec T[s_{j+1}+2\tau-1]$.
  In either case, we have $\LCE(s_i,s_j) = s_{i+1}-s_{i}+2\tau-1$, which concludes the proof.
\end{proof}
\end{full}

We are now ready to describe the complete query algorithm determining
$\LCE(i,j)$ for two positions $i,j$ in $T$.
We start by using \cref{prop:packed} to compare the first $3\tau$ symbols of $T[i\dd n]$ and $T[j\dd n]$.
If we detect a mismatch, the procedure is completed.
Otherwise, we compute the indices $i',j'\in [1\dd n'+1]$ of $s_{i'}=\suc_\S(i)$ and $s_{j'}=\suc_\S(j)$
using rank queries on the bitvector $B$.
If $s_{i'}-i \ne s_{j'}-j$, then we answer the query $\LCE(i,j)=\min(s_{i'}-i, s_{j'}-j)+2\tau-1$ according to \cref{lem:lce}.
Otherwise, \cref{lem:lce} yields $\LCE(i,j)=s_{i'}-i + \LCE(s_{i'},s_{j'})$, and it remains to compute $\LCE(s_{i'},s_{j'})$.
For this, we query for $\ell = \LCE_{T'}(i',j')$ 
and note that $\LCE(s_{i'},s_{j'})=s_{i'+\ell}-s_{i'} + \LCE(s_{i'+\ell},s_{j'+\ell})$ by \cref{lem:lce2}.
Finally, we are left with determining the latter LCE value. 
We start by comparing the first $3\tau$ symbols of $T[s_{i'+\ell}\dd n]$ and $T[s_{j'+\ell}\dd n]$.
If we detect a mismatch, the procedure is finished.
Otherwise, we compute $\LCE(s_{i'+\ell},s_{j'+\ell})=\min(s_{i'+\ell+1}-s_{i'+\ell},s_{j'+\ell+1}-s_{j'+\ell})+2\tau-1$
according to \cref{lem:lce2}. This completes the algorithm.

Before we conclude, note that given a synchronizing set of size $\Oh(\frac{n}{\tau})$ for $\tau = \Oh(\log_\sigma n)$,
the data structure can be constructed in $\Oh(\frac{n}{\tau})$ time.
This follows from \cref{thm:sort-s} (building $T'$), \cref{prop:lce} (LCE queries in $T'$),
and \cref{prp:rksel} ($\rank_B$ queries).
If $\tau \le \eps\log_{\sigma }n $ for a positive constant $\eps < \frac15$, then \cref{lem:sync_fast}~also lets us compute
an appropriate $\tau$-synchronizing set in $\Oh(\frac{n}{\tau})$ time.
The overall construction time, $\Oh(\frac{n}{\tau})$, is minimized by $\tau=\Theta(\log_{\sigma}n)$.

\begin{theorem}
  \label{thm:construct-lce}
  LCE queries in a text  $T\in[0\dd\sigma)^n$ with $\sigma= n^{\bigO(1)}$
  can be answered in $\Oh(1)$ time after $\Oh(n / \log_{\sigma} n)$-time preprocessing of the packed representation of $T$.
\end{theorem}

\section{BWT Construction}
\label{sec:bwt}

Let $T\in[0\dd \sigma)^n$, for $\log \sigma\leq \sqrt{\log n}$, be a text
given in the packed representation, and let $\S$ be a
$\tau$-synchronizing set of $T$ of size $\bigO(\frac{n}{\tau})$, where
$\tau=\eps\log_{\sigma}n$ for some sufficiently small constant
$\eps>0$. We assume that $\tau$ is a positive integer and that $3\tau-1 \leq n$.

In this section, we show how to construct the BWT of $T$ in $\bigO(n \log
\sigma /\sqrt{\log n})$ time and $\bigO(n /\log_\sigma n)$ space.
For simplicity, we first restrict ourselves to a binary
alphabet.  The time and space complexities then simplify~to~$\bigO(n/\sqrt{\log n})$ and $\bigO(n/\log n)$.

\subsection{Binary Alphabet}
\label{sec:bwt-binary-alphabet}

Similarly as in previous sections, we first assume $\R = \emptyset$ (note that this implies $\S \neq \emptyset$ due to $3\tau - 1 \leq n$).
In
\cref{sec:bwt-general-case}, we consider the general case and describe the
remaining parts of our construction.

\subsubsection{The Nonperiodic Case}
\label{sec:bwt-simple-case}

To compensate for the lack of a sentinel $T[n]=\$$
(see \cref{sec:prelim}), let us choose $b_{\$}\in\{0,1\}$ such
that $\per(X)>\frac13\tau$ holds for $X=b_{\$}T[1\dotdot 2\tau)$.
Using packed string matching~\cite{DBLP:journals/tcs/Ben-KikiBBGGW14}, we add 
to $\S$ all positions where $X$ occurs in $T$.
This increases $|\S|$
by $\bigO(\frac{n}{\tau})$ and does not violate~\cref{def:sss}.
Denote by $(s_{i}')_{i\in [1\dd |\S|]}$ the set $\S$,
sorted using \cref{thm:sort-s} according to the order of the
corresponding suffixes, i.e., $T[s_{i}'\dotdot
n]\prec T[s_{j}'\dotdot n]$ if $i<j$. Define a
sequence $W$ of length $|\S|$ so that $W[i]\in[0\dd
2^{3\tau})$ is an integer whose base-2 representation is
$\revbits{T[s_i'-\tau\dotdot s_i'+2\tau)}$, where $\revbits{X}$  
denotes the string-reversal operation.\footnote{Whenever
$T[k]$ is out of bounds, we let $T[k]=b_{\$}$ if
$k=0$ and $T[k]=0$ otherwise.}  In the
word RAM model with word size $w=\Omega(\log n)$, reversing any
$\Oh(\log n)$-bit string takes $\bigO(1)$ time after
$\bigO(n^\delta)$-time ($\delta < 1)$ preprocessing.  Thus,
$W[1\dotdot |\S|]$ can be constructed in
$\bigO(|\S|+n^{\delta})=\bigO(n/\log n)$ time.

Recall that the density condition simplifies to
$\S \cap {[i\dd i+\tau)} \neq \emptyset$ for $i \in [1\dd
n-3\tau+2]$ if $\R=\emptyset$.  Thus, except for $\bigO(\tau)$ rightmost
symbols, every symbol of $T$ is included in at least one character
of~$W$.  In principle, it suffices to rearrange these bits to obtain
the BWT.  For this, we utilize as a black box the wavelet tree of $W$ and prove that its construction performs the necessary
permuting. We are then left with a task of copying the bits from the wavelet tree.

More precisely, we show how to extract (almost) all bits of the BWT of $T$ from
the bitvectors $B_X$ in the wavelet tree of $W$ in $2^{\Theta(\tau)}+\bigO(n/\log n)$
time. Intuitively, we partition the BWT into $2^{\Theta(\tau)}$
blocks that appear as bitvectors $B_X$.

A similar string $W$ was constructed in \cite{DBLP:journals/algorithmica/ChienHSTV15}
for an evenly distributed set of positions. In that case, however, the bitvectors in the wavelet tree
form non-contiguous subsequences of the BWT.
\paragraph{Distinguishing Prefixes}
To devise the announced partitioning of the BWT, for $j\in [1\dd \max \S]$
let  $D_j=T[j\dotdot \suc_\S(j)+2\tau)$ be the \emph{distinguishing prefix} of $T[j\dd n]$.  
By the density condition for $\R=\emptyset$, $\suc_\S(j)-j < \tau$ and thus $|D_j|\leq 3\tau-1$.
Let
$\mathcal{D}=\{D_j : j \leq \max \S\}$.

Recall that $B_X$ is the bitvector associated with the node $v_X$ whose
root-to-node label in the wavelet tree of $W$ is $X$. By
definition of the wavelet tree (applied to $W$), for any $X \in
\{0,1\}^{\leq3\tau-1}$, $B_{\revbits{X}}$ contains the bit preceding $X$ from each string $T[s_i'-\tau\dotdot
s_i'+ 2\tau)$ that has $X$ as a suffix. (The order of these bits matches the sequence $(s'_i)_{i\in [1\dd |\S|]}$.)

\begin{lemma}\label{lm:dp-local-consistency}
  \begin{enumerate}[label=\rm(\arabic*)]
  \item\label{it:dpc1} If $T[j\dotdot j+|X|)=X$ for $X\in \mathcal{D}$, then $D_j=X$.
  \item\label{it:dpc2} If $\SA[b\dotdot e]$ includes all suffixes of $T$ having $X\in \mathcal{D}$
  as a prefix, then $\BWT[b\dotdot e]=B_{\revbits{X}}$.
  \end{enumerate}
\end{lemma}
\begin{proof}
  \ref{it:dpc1} Let $X=D_i$, i.e., $X = T[i\dd \suc_\S(i)+2\tau)$.
  Since $X$ also occurs at position $j$ and $\suc_\S(i)-i\le |X|-2\tau$, we have $\suc_\S(j)-j = \suc_\S(i)-i$
  by \cref{fct:cons}. Consequently, $D_j = T[j\dd \suc_\S(j)+2\tau)=X$.

  \ref{it:dpc2} By the above discussion, $B_{\revbits{X}}$ contains the bits
  preceding $X$ as suffixes in $\smash{(T[s_i'-\tau\dotdot
  s_i'+2\tau))_{i\in [1\dd |\S|]}}$. From \ref{it:dpc1}, there is a bijection
  between the occurrences of $X$ in $T$ and such suffixes
  (importantly, $b_{\$}T[1\dotdot s_i'+2\tau)$ for $s_i'\in \S$ is not a suffix of $X$
  due to the modification of $\S$, so $X$ is never compared against out-of-bounds symbols of $T$ in $W$).
  By definition of $\BWT[b\dotdot e]$ and $|X|\leq
  3\tau - 1$, $B_{\revbits{X}}$ and $\BWT[b\dotdot e]$ indeed contain
  the same (multisets of) bits.

  To show that the bits of $\BWT[b\dotdot e]$ occur in $B_{\revbits{X}}$
  in the same order, observe that $T[s_{i}'+2\tau-|X|\dotdot n] \prec
  T[s_j'+2\tau-|X|\dotdot n]$ holds if $T[s_{i}'-\tau\dotdot
  s_{i}'+2\tau)$ and $T[s_{j}'-\tau\dotdot      
  s_{j}'+2\tau)$ have $X$ as a suffix for $i<j$.
  This is because
  $T[s_{i}'+2\tau-|X|\dd s_{i}')=T[s_{j}'+2\tau-|X|\dd s_{j}')$
  is a prefix of $X$, and 
  we have $T[s_{i}'\dotdot n] \prec T[s_{j}'\dotdot
  n]$ by $i < j$.
\end{proof}

\paragraph{The Algorithm}
We start by building the string $W$ and its wavelet tree.  By
\cref{thm:wavelet_tree}, this takes $\bigO((n/\tau) \allowbreak
\log(2^{3\tau}) / \sqrt{\log(n/\tau)}) = \bigO(n/\sqrt{\log n})$ time
and $\bigO(n/\log n)$ space.

Next, we create a lookup table that, for any
$X\in\{0,1\}^{2\tau}$, tells whether $X$ occurs at a position $j\in \S$
(by the consistency condition, $j\in \S$ for every position $j$ where $X$ occurs in $T$). 
It needs $\bigO(n^{2\eps})$ space and is easily
filled: set ``yes'' for each $T[j\dotdot j+2\tau)$ with $j\in \S$.

Initialize the output to an empty string. Consider the preorder
traversal of a complete binary tree of depth $3\tau-1$ with each edge
to a left child labeled ``0'' and to a right child---``1''.
Whenever we visit a node with root-to-node path $X$ such that $|X|\ge 2\tau$, we check if the length-$2\tau$ suffix of $X$ is a
``yes'' substring.  If so, we report $X$ and skip the traversal of the current
subtree. Otherwise, we descend into the subtree. 
This procedure enumerates $\mathcal{D}$ in the lexicographic
order in $\bigO(2^{3\tau}) = \bigO(n^{3\eps})$ time. For each reported
substring $X$, we append $B_{\revbits{X}}$ to the output
string. Locating $v_{\revbits{X}}$ takes $\bigO(|X|)=\bigO(\log n)$
time; hence, we spend $\bigO(n/\log n + n^{3\eps}\log
n)=\bigO(n/\log n)$ time in total.

The above traversal outputs a BWT subsequence containing the symbols preceding
positions in $[1\dotdot \max \S]$. To include the missing symbols, we
make the following adjustment: while visiting a node with label $X\notin \mathcal{D}$,
we check if $X$ occurs as a suffix of $T$. If so, then before descending into the subtree, we append the preceding character $T[n-|X|]$ 
to the output string.

The algorithm runs in $\bigO(n/\sqrt{\log n})$ time and uses
$\bigO(n/\log n)$ space. For correctness, observe that the set $\mathcal{D}$ is
prefix-free by \cref{lm:dp-local-consistency}.  Thus, no symbol is
output twice. 

To complete the construction, we need $\SA^{-1}[1]$ (see
\cref{sec:sa}). 
Let $D_1$ be the distinguishing prefix of $T[1\dd n]$ and let $i_1$ be the index of $\min \S$
in $(s_i')_{i\in[1\dotdot |\S|]}$. Observe that the symbol $T[0]=b_\$$ occurs in $B_{\revbits{D_1}}$
at position $|\{i\le i_{1} : \revbits{D_1} \text{ is a prefix of }W[i]\}|$, which can be determined in $\Oh(|\S|)$ time.
Appending $B_{\revbits{D_1}}$ to the constructed BWT, 
we map this position in $B_{\revbits{D_1}}$ to the corresponding one in the BWT.
Finally, we overwrite $b_{\$}$ by setting $\BWT[\SA^{-1}[1]]=T[n]$.

\subsubsection{The General Case}
\label{sec:bwt-general-case}
Let use define
\[\F = \{X\in \{0,1\}^{3\tau-1} : X'\notin \D \text{ for every prefix }X'\text{ of }X\}.\]
Observe that if $T[j\dotdot j+3\tau-1)\in \F$, then $j\in \R$.  Conversely, whenever
$j\in \R$, then $T[j\dotdot j+3\tau-1)\in \F$.  Thus,    
$\R$ contains precisely the starting positions of all strings in
$\F$. Hence, in the general case with $\R$ possibly non-empty,
 the algorithm of \cref{sec:bwt-simple-case} outputs
the BWT subsequence missing exactly the symbols $T[j-1]$ for~$j\in\R$.

The crucial property of $\R$ that allows handling the general case is
that $\R$ cannot have many ``gaps''. Moreover, whenever
$X\in\F$ occurs at a position
$j\in\R$ with $j-1\in\R$, then $T[j-1]$ depends~on~$X$.

\begin{lemma}
  \label{lm:properties-of-Rprim}
  Let $\R'=\{j\in\R : j-1\notin\R\}$ be a subset of $\R$. Then:
  \begin{enumerate}[label=\rm(\arabic*)]
  \item\label{it:pr1} $|\R'| \leq |\S|+1$.
  \item\label{it:pr2} If $X\,{=}\,T[j\dotdot j+|X|)\,{\in}\,\F$ and
    $j\,{\notin}\,\R'$, then $T[j-1]\,{=}\,X[\per(X)]$.
  \end{enumerate}
\end{lemma}
\begin{proof}
\ref{it:pr1} By density condition, $j\in \R'$ implies $j-1\in \S$ if $j>1$.

\ref{it:pr2} Note that $\per(X)=\per(T[j-1\dd \suc_\S(j-1)+2\tau-2])\le\frac13\tau$
due to \cref{fct:den} and because $\suc_\S(j-1)=\suc_\S(j)\ge j+\tau$.
Hence, $T[j-1]=T[j-1+\per(X)]=X[\per(X)]$.
\end{proof}

Consider thus the following modification: whenever we reach
$X\in\F$ during the enumeration of $\mathcal{D}$, we append to the output a
unary string of $f_X$ symbols $X[\per(X)]$, where $f_X$ is the number
of occurrences of $X$ in $T$. By \cref{lm:properties-of-Rprim},
the number of mistakes in the resulting BWT, over all
$X\in\F$, is only $|\R'|=\bigO(\frac{n}{\tau})$.

To implement the above modification (excluding BWT correction), we
need to compute $\per(X)$ and $f_X$ for every
$X\in\{0,1\}^{3\tau-1}$. The period is determined using a lookup table.

\paragraph{Computing Frequencies of Length-$\ell$ Substrings}
Consider $\lfloor|T|/\ell\rfloor$ blocks of length $2\ell-1$ starting in $T$ at
positions of the form $1+k\ell$ (the last block might be shorter). 
Sort all blocks in $\Oh(\frac{n}{\ell})$ time into a list $L$. Then, scan the list and for each \emph{distinct}
block $B$ in $L$, consider the multiset of all its length-$\ell$ substrings $X$. For each such $X$, increase its frequency by the
frequency of $B$ in $L$.
The correctness follows by noting that $T[i\dd i+\ell)$ 
is contained in the $\ceil{i/\ell}$th block only.

There are at most $1+ 2^{2\ell-1}$ distinct blocks and we spend $\bigO(\ell)$ time
for each. The total running time is therefore 
$\Oh(\frac{n}{\ell}+\ell 2^{2\ell-1})$. 

In our application, $\ell=3\tau-1 < 3\eps\log  n$, so it suffices to choose
 $\eps < \frac16$ so that $\bigO(\frac{n}{\ell}+\ell 2^{2\ell-1}) =
\bigO(n/\log n + n^{6\eps}\log n) = \bigO(n/\log n)$.

\paragraph{Correcting BWT}
We will now show how to compute the rank (i.e., the position in the suffix array) of every suffix of $T$ starting in $\R'$. This will
let us correct the mistakes in the $\BWT$ produced within the
previous step. If $r_j$ is the rank of $T[j\dotdot n]$, where
$j\in\R'$, we set $\BWT[r_j] = T[j-1]$.  To compute $r_j$, we only
need to know $r_j'$: the \emph{local rank} of $T[j\dotdot n]$ among the
suffixes of $T$ starting with $T[j\dotdot j + 3\tau - 1)$
(note that any such suffix starts at $j'\in \R$) since the rank among other suffixes is known
during the enumeration of $\mathcal{D}$.  Formally, for $j\in\R'$, define
$\pos(j) = \{j'\in\R : \LCE_T(j,j')\geq3\tau-1 \text{ and }
T[j'\dotdot n]\preceq T[j\dotdot n]\}$ so that~$r_j'=|\pos(j)|$.

Motivated by \cref{lm:properties-of-Rprim}, we focus on the
properties of \emph{runs} of consecutive positions in $\R$.  We start by
partitioning such runs into classes, where the
computation of local ranks is easier and can be done independently.
For $X\in\F$, we define the \emph{Lyndon root}
$\Lroot(X)=\min\{X[t\dotdot t+p):t\in[1\dotdot p]\}\in\{0,1\}^{\leq
\tau/3} $, where $p=\per(X)$. We further set
$\Lroot(j)=\Lroot(T[j\dotdot j+3\tau-1)) $ for every $j\in\R$.  It is
easy to see that if $j\in\R\setminus\R'$, then
$\Lroot(j-1)=\Lroot(j)$.  Thus, to compute $r_j'$ for some
$j\in\R'$, it suffices to look at the runs starting at $j'\in\R'$
such that $\Lroot(j)=\Lroot(j')$.

Further, for $j\in\R$, let us define $e_{j}=\min\{j' \geq j : j'\not\in\R\}+3\tau-2$.
 We define $\type(j)=+1$ if
$T[e_j]\succ T[e_j-p]$ and $\type(j)=-1$ otherwise, where
$p=\per(T[j\dotdot e_j))$.  Similarly as for the L-root, if
$j\in\R\setminus\R'$, then $\type(j-1)=\type(j)$.  Furthermore, if
$\type(j)=-1$ holds for $j\in\R'$, then $\type(j')=-1$ holds for all
$j'\in\pos(j)$.  Let
$\R^{-}=\{j\in\R : \type(j)=-1\}$, $\R^{+}=\R\setminus \R^{-}$,
$\R'^{-}=\R'\cap \R^{-}$, and $\R'^{+}=\R'\cap\R^{+}$.
In the rest of this section, we
focus on computing $r_j'$ for $j\in\R'^{-}$. The set $\R'^{+}$ is
processed~symmetrically.

To efficiently determine local ranks for a group of runs with the
same L-root, we refine the classification further into individual
elements of $\R$.  Let $U=\Lroot(j)$ for $j\in\R$.  It is easy to see
that the following \emph{L-decomposition} $T[j\dotdot e_j)=U'U^kU''$
(where $k\geq 1$,  $U'$ is a proper suffix  
of $U$, and $U''$ is a proper prefix of $U$) is unique.  We call the triple $(|U'|, k,
|U''|)$ the \emph{L-signature} of $j$ and the value $\Lexp(j)=k$ its
\emph{L-exponent}. By the uniqueness of L-decompositions, given
$j\in\R'^{-}$, we have $\Lexp(j')\leq \Lexp(j)$ for all
$j'\in\pos(j)$.

Note that, letting $\S=(s_i)_{i\in[1\dotdot |\S|]}$ where
$s_i<s_{i'}$ if $i<i'$ and $s_0=0$,
$s_{|\S|+1}=n-2\tau+2$, we have, by the density condition: $\R'=\{s_i+1 :
i\in[0\dotdot |\S|]\text{ and }s_{i+1}\,{-}\,s_{i}\,{>}\,\tau\}$.
Furthermore, whenever $j-1=s_i$ for $j\in\R'$, then
$e_j=s_{i+1}+2\tau-1$. Thus, computing $\R'$ and $\{e_j\}_{j\in\R'}$ (and
also the type of each $j\in\R'$) takes $\bigO(|\S|)$ time.

For $j\in\R'^{-}$, let %
\begin{full}%
  \[ r_j^{=}=|\{j'\in\pos(j) : \Lexp(j') = \Lexp(j)\}| \qquad\text{and}\qquad r_j^{<}=|\{j'\in\pos(j) : \Lexp(j') < \Lexp(j)\}|.\]
\end{full}
\begin{camera}
$r_j^{=}=|\{j'\in\pos(j) : \Lexp(j') = \Lexp(j)\}|$ and
$r_j^{<}=|\{j'\in\pos(j) : \Lexp(j') < \Lexp(j)\}|$.
\end{camera}

To compute $r_j^{=}$ for each $j\in\R'^{-}$, consider sorting the list of
all $j\in\R'^{-}$ first according to $\Lroot(j)$, second
according to $|U''|$ in its L-signature, and third according to
$T[e_j\dotdot n]$. Such ordering can be obtained in $\bigO(|\R'^{-}|)$
time by utilizing the sequence $(s_i')_{i\in[1\dotdot |\S|]}$ and the
fact that $e_j-2\tau+1\in \S$. The L-root and $|U''|$ are computed using lookup
tables. Then, to determine $r_j^{=}$, we count $j'\in\R'^{-}$ with the
same L-root that are not later than $j$ in the list and the factor
$U'U^k$ in their L-decomposition is at least as long as for $j$. To this end,
we issue a 3-sided orthogonal 2D range counting query on an input
instance containing, for every $j\in\R'^{-}$, a point with $|U'U^{k}|$
from its L-decomposition as the first coordinate and its position in the
above list as the second coordinate.  Answering a batch of $m$
orthogonal 2D range counting queries takes $\bigO(m\sqrt{\log m})$
time and $\Oh(m)$ space~\cite{ChanP10}. Since $m=|\R'^{-}|$, this step takes
$\bigO(n/\sqrt{\log n})$ time.

To compute $r_j^{<}$ for each $j\in\R'^{-}$, we sort $j\in\R'^{-}$
first by $\Lroot(j)$, and then by $\Lexp(j)$.  For a fixed
$\smash{U\in\{0,1\}^{\leq \tau/3}}$, let us define
\[\R^{-}_U=\{j\in\R^{-} : \Lroot(j)=U\} \subseteq \R^{-}.\]
Let also $\R'^{-}_U=\R'\cap\R^{-}_U$. The key observation is that there are only $|U|$ different prefixes of
length $3\tau-1$ in $\{T[j\dotdot n] : j\in\R^{-}_U\}$.  We will incrementally compute the frequency of each of these
prefixes $X\in\{0,1\}^{3\tau-1}$ and keep the count in an array
$C[0\dotdot |U|)$, indexed by $t=|U'|$ in the L-decomposition of every
$j\in\R^{-}_U$ with $T[j\dotdot j+3\tau-1)=X$ (denote this set    
as $\R^{-}_{U,t}$).  A single round of the algorithm handles
$\mathsf{H}_k=\{j\in\R'^{-}_U : \Lexp(j)=k\}$.  We execute the rounds
for increasing $k$ and maintain the invariant that $C[t]$ contains $|\{j\in\R^{-}_{U,t} : \Lexp(j) \leq
k\}|$ at the end of round $k$.  At the beginning of round $k$, we use the values of $C$ to
first compute $r_{j}^{<}$ for each $j\in \mathsf{H}_k$ and then update
$C$ to maintain the invariant.  The update consists of increasing some
entries in $C$ for each $j\in \mathsf{H}_k$, and then increasing all
entries in $C$ by the total number $q$ of yet unprocessed positions
(having higher L-exponent), i.e.,
$q=|\R'^{-}_U|-\sum_{i=1}^{k}|\mathsf{H}_i|$.  It is easy see that for
each $j\in\mathsf{H}_k$, the update can be expressed as a constant
number of increments in contiguous ranges of $C$.  Additionally, if
$\mathsf{H}_{k-1}=\emptyset$ and $k_{\rm prev}=\max\{k'<k :
\mathsf{H}_{k'}\neq\emptyset\}$, right at the beginning of round $k$
(before computing $r_j^{<}$ for $j\in\mathsf{H}_k$), we
increment all of $C$ by $q\cdot (k-k_{\rm prev}-1)$, to account for
the skipped L-exponents.  Each update of $C$ takes $\bigO(\log
|U|)=\bigO(\log \log n)$ time if we implement $C$ as a balanced
BST. Thus, the algorithm takes $\bigO(|\R'^{-}| \log \log n) =
\bigO(n/\sqrt{\log n})$ time. Note, that there are
$\bigO(2^{\tau/3})=\bigO(n^{\eps/3})$ different L-roots; hence, we can afford
to initialize $C$ in $\bigO(\log n)$ time for each $U$.

\begin{theorem}
  \label{thm:construct-bwt}
  Given the packed representation of a text ${T\,{\in}\,\{0,1\}^n}$ and its
  $\tau$-synchronizing set $\S$ of size $\bigO(n/\tau)$ for $\tau =
  \eps\log n$, where $\eps > 0$ is a sufficiently small constant, 
  the Burrows--Wheeler transform of $T$ can be constructed in
  $\bigO(n/\sqrt{\log n})$ time and $\bigO(n/\log n)$ space.
\end{theorem}

\subsection{Large Alphabets}
\label{sec:large-alphabet}

Note that our BWT construction does not immediately generalize to
larger alphabets since for binary strings it already relies on wavelet
tree construction for sequences over alphabets of polynomial size.
More precisely, in the binary case, we combined the bitvectors $B_X$ in
the wavelet tree of a large-alphabet sequence $W$ to retrieve
fragments of the Burrows--Wheeler transform of the text $T$.  For
an alphabet $\Sigma=[0\dd \sigma)$, the BWT consists of $(\log\sigma)$-bit characters.  
Thus, instead of standard
binary wavelet trees, we use wavelet trees of degree $\sigma$. For
simplicity, we assume that $\sigma$ is a power of two.

\subsubsection{High-Degree Wavelet Trees}
\label{sec:high-degree-wt}

To construct a degree-$\sigma$ wavelet tree,
we consider a string $W$ of length $n$ over an alphabet
$[0\dd\sigma^b)$ and think of every symbol $W[i]$ as of a
number in base $\sigma$ with exactly $b$ digits (including leading
zeros).  First, we create the root node $v_{\eps}$ and construct its string
$D_\eps$ of length $n$ setting as $D_\eps[i]$ the most significant digit of
$W[i]$.  We then partition $W$ into $\sigma$ subsequences
$W_{0},W_{1},\ldots,W_{\sigma-1}$ by scanning through $W$ and
appending $W[i]$ with the most significant digit removed to $W_{c}$,
where $c$ is the removed digit of $W[i]$. We recursively repeat the
construction for every $W_{c}$ and attach the resulting tree as the
$c$th child of the root.  The nodes of the resulting degree-$\sigma$
wavelet tree are labeled with strings $Y\in [0\dd\sigma)^{\le
  b}$.  For $|Y| < b$, the string $D_Y$ at node $v_Y$ labeled with $Y$
contains the next digit following the prefix $Y$ in the $\sigma$-ary
representation of $W[i]$ for each element $W[i]$ of $W$ whose
$\sigma$-ary representation contains $Y$ as a prefix. (The digits in $D_Y$ occur in the order as the corresponding entries $W[i]$.)
The total size of a wavelet tree is $\Oh\big(\sigma^b + \frac{n b\log\sigma}{\log n}\big)$
words, which is $\Oh\big(\frac{n b\log\sigma}{\log n}\big)$ if $b \le \log_\sigma n$.

As observed in \cite{WaveletSuffixTree}, the binary wavelet tree
construction procedure can be used as a black box to build degree-$\sigma$
wavelet trees. Here, we present a more general version of this
reduction.

\begin{camera}
  \begin{lemma}[\fullonly; see {\cite[Lemma 2.2]{WaveletSuffixTree}}]
    \label{lm:rebuilding}
    Given the packed representation of a string $\smash{W \in \big[0\dd \sigma^b\big)^n}$
    with $b \le \log_\sigma n$, we can
    construct its degree-$\sigma$ wavelet tree in $\Oh(n b\log\sigma /
    \sqrt{\log n}+ nb \log^2 \sigma / \log n)$ time using $\Oh(n b / \log_\sigma n)$ space.
  \end{lemma}
\end{camera}

\begin{full}
\begin{lemma}[see {\cite[Lemma 2.2]{WaveletSuffixTree}}]
  \label{lm:rebuilding}
  Given the packed representation of a string $\smash{W \in \big[0\dd \sigma^b\big)^n}$
  with $b \le \log_\sigma n$, we can
  construct its degree-$\sigma$ wavelet tree in $\Oh(n b\log\sigma /
  \sqrt{\log n}+ nb \log^2 \sigma / \log n)$ time using $\Oh(n b / \log_\sigma n)$ space.
\end{lemma}
\begin{proof}
  Consider a binary wavelet tree for $W$, where each symbol $W[i]$ is
  interpreted as an integer with $b\log\sigma$ binary digits.  For
  every node $v_X$, where $X\in\{0,1\}^{< b\log\sigma}$, we define
  $D_X$ as the bitvector containing the $\log\sigma$ bits following
  $X$ in every symbol $W[i]$ prefixed with $X$ (in the order these symbols
  appear in $W$). If $|X|>(b-1)\log \sigma$, we pad $W[i]$ with trailing
  zeros. Bitvectors $D_X$ can be interpreted as strings
  over $[0\dd \sigma)$.  In the rest of the proof, we adopt the
  latter convention. This way, we have generalized the strings
  $D_Y$, i.e., for every node $v_Y$ from the degree-$\sigma$ wavelet
  tree of $W$, where $Y\in[0\dd \sigma)^{<b}$, we have $D_Y=D_X$,
  where $X$ is the binary encoding of $Y$.

  The algorithm starts by constructing the binary wavelet tree for
  $W$. Next, we compute the strings $D_X$ for all its nodes.  To obtain the
  degree-$\sigma$ wavelet tree of $W$, it then suffices to remove all
  nodes whose depth is not a multiple of $\log \sigma$. For each
  surviving node, we set its nearest preserved ancestor as the
  parent. Each inner node has $\sigma$ children, and we order them
  consistently with the left-to-right order in the original binary wavelet tree.

  The key observation is that for any $X\in\{0,1\}^{<b\log\sigma-1}$,
  the string $D_X$ can be computed from $D_{X0}$, $D_{X1}$, and $B_X$. First, we
  shift the binary representations of the characters from $D_{X0}$ and
  $D_{X1}$ by a single bit to the right, prepending a $0$ to the
  characters in $D_{X0}$ and a $1$ to the characters in
  $D_{X1}$. Then, we construct $D_{X}$ by interleaving $D_{X0}$ and
  $D_{X1}$ according to the order defined by $B_{X}$: if the $i$th bit
  is $0$, we append to the constructed string
  $D_{X}$ the next character of $D_{X0}$, and otherwise we append the next character of $D_{X1}$.

  We pack every $\frac{1}{4}\frac{\log n}{\log \sigma}$ consecutive
  characters of strings
  $D_X$ with $X\in\{0,1\}^{<b\log\sigma}$ into a single machine
  word.  During interleaving, instead of accessing $D_{X0}$ and
  $D_{X1}$ directly, we keep two buffers, each of at most
  $\frac{1}{4}\frac{\log n}{\log \sigma}$ yet-unmerged characters from
  the corresponding string, as well as a buffer of at most
  $\frac{1}{4}\log n$ yet-unused bits of $B_X$. We continue the
  computation of $D_X$ until one of the buffers becomes empty.  To
  implement this efficiently, we preprocess (and store in a lookup
  table) all possible scenarios between two buffer refills for every
  possible initial content of the input buffers. We store the
  generated data (of at most $\frac{1}{2}\log n$ bits) and the final
  content of all buffers.
  
  The preprocessing takes $\tilde\Oh(2^{\frac{3}{4}\log n})=o(n/\log n)$ time
  and space. The number of operations required to generate $D_{X}$ is
  proportional to the number of times we reload the buffers;
  hence, it takes $\Oh(1+|D_X|/\frac{\log n}{\log\sigma})$ time.  Due to
  $\smash{\sum_{X\in\{0,1\}^{<b\log\sigma}}|D_X|=n(b\log\sigma-1)}$,
  the total complexity is $\smash{\Oh(\sigma^b+nb \log^2{\sigma} /
  \log n)=\bigO(nb\log^2{\sigma}/\log n)}$.  Adding the time to
  construct the binary wavelet tree of $W$ (see \cref{thm:wavelet_tree}),
  this yields the final complexity.
\end{proof}
\end{full}

\subsubsection{BWT Construction Algorithm}
\label{sec:large-alphabet-bwt-constr}

Our construction algorithm for $T\in[0\dotdot \sigma)^n$
uses a $\tau$-synchronizing set
$(s_i')_{i\in[1\dotdot|\S|]}$ (where $T[s_{i}'\dotdot n] \prec
T[s_{j}'\dotdot n]$ if $i<j$) for $\tau = \eps \log_{\sigma} n$.
We then build a sequence $W\in [0\dd\sigma^{3\tau})^{|\S|}$ with $W[i] = \revbits{T[s_i'-\tau\dotdot
s_i'+2\tau)}$, i.e., we reverse 
$T[s_i'-\tau\dotdot s_i'+2\tau)$ and then interpret it as a
$3\tau$-digit integer in base $\sigma$.  Next, we construct a      
degree-$\sigma$ wavelet tree of $W$ and combine the strings $D_Y$,
where $Y\in [0\dd \sigma)^{\leq 3\tau-1}$, to obtain the BWT of $T$. The
procedure is analogous to the binary case.

\begin{theorem}
  \label{thm:construct-bwt-large}
  Given the packed representation of a text $T\in[0\dotdot \sigma)^n$
  with $\log \sigma\le \sqrt{\log n}$ and its $\tau$-synchronizing set
  $\S$ of size $\bigO(\frac{n}{\tau})$ with $\tau = \eps\log_\sigma n$, where
  $\eps > 0$ is a sufficiently small constant, the
  Burrows--Wheeler transform of $T$ can be constructed in $\bigO(n\log \sigma/\sqrt{\log
  n})$ time and $\bigO(n/\log_\sigma n)$ space.
\end{theorem}

\section{Conditional Optimality of the Binary BWT Construction}
\label{sec:conditional-optimality}

Given an array $A[1\dotdot m]$ of integers, the task of counting
inversions is to compute the number of pairs $(i,j)$ such that $i < j$ and $A[i] > A[j]$.  The
currently fastest solution for the above problem, due to Chan and
P\v{a}tra\c{s}cu~\cite{ChanP10}, runs in $\bigO(m\sqrt{\log m})$ time
and $\bigO(m)$ space.

Without loss of generality~\cite{Han04}, we assume
$A[i]\in[0\dotdot m)$.  In this section, we show that improving the BWT
construction from~\cref{sec:bwt-binary-alphabet} also yields an
improvement over~\cite{ChanP10}. More precisely, we show that
computing the BWT of a packed text $T\in\{0,1\}^n$ in time $\bigO(f(n))$ implies an
$\bigO(m+f(m \log m))$-time construction of the wavelet tree for $A$; hence,
improving over~\cref{thm:construct-bwt} implies an $o(m\sqrt{\log
  m})$-time wavelet tree construction.  The main result then follows
from the next observation since it is easy to count inversions
in a length-$n$ bitvector in $\bigO(n/\log n)$ time.

\begin{observation}
  The number of inversions in any integer sequence $A$ is equal to the
  total number of inversions in the bitvectors of the wavelet tree of $A$.
\end{observation}

We first consider the case when $A[i] < m^\eps$ for all $i\in[1\dotdot m]$ and for some sufficiently small constant
$\eps<1$; it is almost a direct reversal of the reduction from
\cref{sec:bwt-simple-case}.

\subsection{The Case \texorpdfstring{$A[i] < m^{\eps}$}{A[i]<m-to-eps}}
\label{sec:lower-bound-simple-case}

Let $\text{\rm bin}_{k}(x)\in\{0,1\}^{k}$ be the base-2 representation
of $x\in[0\dotdot 2^k)$ and let $\text{\rm pad}_{k} : \{0,1\}^{k}
\rightarrow \{0,1\}^{2k}$ be a \emph{padding} function that, given
$X\in\{0,1\}^k$, inserts a 0-bit before each bit of $X$.

Assume $\eps \log m$ and $\log m$ are integers.  Given $A[1\dotdot
  m]$, let
\[
    T_A=\prod_{i=1}^{m}\revbits{\text{\rm bin}_{\eps\log m}(A[i])}
    \cdot 01\cdot 1^{\eps\log m} \cdot \text{\rm pad}_{\log
      m}(\text{\rm}\text{\rm bin}_{\log m}(i-1)) \cdot 0.
\]
The text $T_A$ is of length $m(3+2(1+\eps)\log m)$ and constructing it 
takes $\bigO(m)$ time.  Recall that $B_X$ denotes the bitvector corresponding to the node of the wavelet tree of
$A$ with root-to-node label $X$.

\begin{lemma}
  For $X\in\{0,1\}^{<\eps\log m}$, let $\SA[b\dotdot e]$ be the range
  containing all suffixes of $T_A$ having $\revbits{X}\cdot 01\cdot 1^{\eps\log
  m}$ as a prefix. Then, $B_X=\BWT[b\dotdot e]$.
\end{lemma}
\begin{proof}
  Let $g : i \mapsto (i-1)(3+2(1+\eps)\log
  m)+\eps\log m-|X|+1$ restricted to
  $A_X:=\{i\in[1\dotdot m] : X\text{ is a prefix of }\text{\rm
  bin}_{\eps\log m}(A[i])\}$ be a one-to-one map between $A_X$ and
  the set $\SA[b\dotdot e]$. It remains to observe that $T_A[g(i)\dotdot
  |T_A|]\prec T_A[g(j)\dotdot |T_A|]$ holds if $i<j$ for $i,j\in
  A_X$, and $\text{\rm bin}_{\eps\log m}(A[i])[|X|+1]=T_A[g(i)-1]$ for
  $i\in A_X$.
\end{proof}

To compute the BWT ranges corresponding to the bitvectors $B_X$ for
$\smash{X\in\{0,1\}^{<\eps\log m}}$, we proceed as in \cref{sec:bwt}:
perform a preorder traversal of the complete binary tree of height
$2\eps \log m+1$ corresponding to lexicographic enumeration of
$\smash{X'\in\{0,1\}^{\leq 2\eps\log m+1}}$. For $\smash{X'=X011^{\eps\log m}}$, where
$\smash{X\in\{0,1\}^{<\eps \log m}}$, copy the bits from BWT into $B_{\revbits{X}}$. The
number of bits to copy is given by the number $f_{X'}$ of occurrences of
$X'$ in $T_A$.  After copying the bits (or
when we reach $\smash{X'\in\{0,1\}^{2\eps\log m+1}}$), advance the
position in BWT by $f_{X'}$.

To determine the frequencies of all $\{0,1\}^{\leq 2\eps\log m+1}$,
compute $f_{X'}$ for all $X'\in\{0,1\}^{2\eps\log m+1}$ in $\bigO(m)$
time using the algorithm in~\cref{sec:bwt-general-case}.  The
remaining frequencies are then easily derived.

The algorithm runs in $\bigO(m+f(m \log m))$ time, where $f(n)$ is the
runtime of the BWT construction for a packed text from $\{0,1\}^{n}$.

\subsection{The Case \texorpdfstring{$A[i] < m$}{A[i]<m}}
\label{sec:lower-bound-general-case}

Given an array $A[1\dotdot m]$ with values $A[i]\in [0\dd m)$, let
\[
    T_A=\prod_{i=1}^{m}\revbits{\text{\rm bin}_{\log m}(A[i])} \cdot
    01 \cdot 1^{\log m}\cdot 0 \cdot\text{\rm}\text{\rm bin}_{\log m}(i-1) \cdot
    0\cdot 1^{\log m}\cdot 0.
\]

\begin{lemma}
  For $X\in\{0,1\}^{<\log m}$, let $\SA[b\dotdot e]$ be the range
  containing all suffixes of $T_A$ having $\revbits{X}\cdot 01\cdot 1^{\log m}$
  as a prefix. Then, $B_X=\BWT[b\dotdot e]$.
\end{lemma}

The main challenge lies in obtaining the $\BWT$ ranges, as the
approach of~\cref{sec:lower-bound-simple-case} no longer works.
Consider instead partitioning the suffixes in $\SA$ according to the
length-$\log m$ prefix (separately handling shorter suffixes). Observe
that there is a unique bijection $\mathrm{ext} : \{0,1\}^{\log m}\setminus\{1^{\log m}\}
\to \{0,1\}^{<\log m}\cdot 01\cdot 1^{\log m}$
such that $X$ is a prefix of $\mathrm{ext}(X)$. Furthermore, if the range
$\SA[b\dotdot e]$ contains suffixes of $T_A$ prefixed by
$X\in\{0,1\}^{\log m}\setminus\{1^{\log m}\}$, then the analogous range
$\SA[b'\dotdot e']$ for $\mathrm{ext}(X)$ satisfies $e'=e$.  Thus,
it suffices to precompute frequencies $F:=\{(X,f_X) :
X\in\{0,1\}^{\log m}\setminus\{1^{\log m}\}\}$ and $F':=\{(X,f_X) :
X\in\{0,1\}^{<\log m}\cdot 01\cdot 1^{\log m}\}$.

To this end, construct $F_{\rm pref}\,{:}{=}\,\{(X,f_X') :
X\in\{0,1\}^{\leq \log m}\}$, with $f_X'$ defined as the number occurrences
of $X$ as prefixes of strings in $\mathcal{A}=(\revbits{\text{\rm
bin}_{\log m}(A[i])})_{i\in[1\dotdot m]}$. To compute $F_{\rm
pref}$, we first in $\bigO(m)$ time build $\{(X,f_X') :
X\in\{0,1\}^{\log m}\}\subseteq F_{\rm pref}$ by scanning
$\mathcal{A}$.~The remaining elements of $F_{\rm pref}$ are derived by
the equality $f_{X}' = f_{X0}' + f_{X1}'$.  Analogously prepare suffix
frequencies $F_{\rm suf}$.  Given $F_{\rm pref}$ and $F_{\rm suf}$, we
can compute $F$ in $\bigO(m)$ time: the number of occurrences of
$X\in\{0,1\}^{\log m}$ overlapping factors $\mathcal{A}$ is obtained
from $F_{\rm pref}$ and $F_{\rm suf}$; the number of other occurrences
is easily determined as it only depends on $m$.  Finally, $F'$ is
computed directly from $F_{\rm suf}$.

\begin{theorem}
  If there exists an algorithm that, given the packed representation of
  a text $T\in\{0,1\}^n$, constructs its $\BWT$ in $\bigO(f(n))$ time,
  then we can compute the number of inversions in an array of $m$
  integers in $\bigO(m+f(m \log m))$ time. In particular, if
  $f(n)=o(n/\sqrt{\log n})$, then the algorithm runs in $o(m\sqrt{\log m})$
  time.
\end{theorem}

\section{Synchronizing Sets Construction}
\label{sec:construction-of-set-S}

Throughout this section, we fix a text $T$ of length $n$ and a positive integer $\tau \le \frac12 n$.
We also introduce a partition $\PP$ of the set $[1 \dd {n-\tau+1}]$ so that positions $i,j$ belong to the same class
if and only if $T[i\dd i+\tau)=T[j\dd j+\tau)$. In other words, each class contains the starting positions
of a certain length-$\tau$ substring of $T$.
We represent $\PP$ using an \emph{identifier function} $\id :  [1\dd n-\tau+1] \to [0\dd |\PP|)$ such that $\id(i)=\id(j)$ if and only if $T[i\dd i+\tau) = T[j\dd  {j+\tau})$.

The local consistency of a string synchronizing set $\S$ means that the decision on whether $i\in \S$
should be made solely based on $T[i\dd i+2\tau)$ or, equivalently, on $\id(i),\ldots, \id(i+\tau)$. 
The density condition, on the other hand, is formulated in terms of a set
$\R = \{i \in [1\dd n-3\tau+2] : \per(T[i\dd i+3\tau-2])\le \frac13 \tau\}$.
Here, we introduce its superset $\Q = \{i \in [1\dd n-\tau+1] : \per(T[i\dd i+\tau) \le \frac13\tau)\}$.
The periodicity lemma (\cref{lem:per}) lets us relate these two sets:
\begin{equation}\label{obs:qr}
\R = \{i \in [1\dd n-3\tau+2] : [i\dd i+2\tau)\sub \Q\}.\end{equation}

Based on an identifier function $\id$ and the set $\Q$,
we define a synchronizing set $\S$ as follows.
Consider a window of size $\tau+1$ sliding over the identifiers $\id(j)$. For
any position $i$ of the window, we compute the smallest identifier $\id(j)$ for $j\in [i\dd i+\tau]\sm \Q$.
We insert $i$ to $\S$ if the minimum is attained for $\id(i)$ or $\id(i+\tau)$.

\begin{camera}
\begin{construction}\label{cons:sync}
  For an identifier function $\id$, we define
  \begin{multline*}\S = \{i\in [1\dd n-2\tau+1] : \\ \min\{\id(j) : j\in [i\dd i+\tau]\sm \Q\} \in \{\id(i),\id(i+\tau)\}\}.
  \end{multline*}
\end{construction}
\end{camera}
\begin{full}
  \begin{construction}\label{cons:sync}
    For an identifier function $\id$, we define
    \[\S = \{i\in [1\dd n-2\tau+1] : \min\{\id(j) : j\in [i\dd i+\tau]\sm \Q\} \in \{\id(i),\id(i+\tau)\}\}.\]
  \end{construction}
\end{full}
Let us formally argue that this results in a synchronizing set.
\begin{lemma}  \label{lm:s-from-sync}
  \cref{cons:sync} always yields a $\tau$-synchronizing set.
\end{lemma}
\begin{proof}
  As for the local consistency of $\S$, observe that if positions $i,i'$ satisfy
  $T[i\dd i+2\tau)=T[i'\dd i'+2\tau)$, then $\id(i+\delta)=\id(i'+\delta)$ for $\delta\in[0\dd \tau]$.
  Moreover, $i+\delta \in \Q$ if and only if $i'+\delta \in \Q$.
  Consequently, $\{\id(i),\id(i+\tau)\}=\{\id(i'),\id(i'+\tau)\}$ and
   \[\min\{\id(j) : j\in [i\dd i+\tau]\sm \Q\}=\min\{\id(j) : j\in[i'\dd i'+\tau]\sm \Q\}.\]
  Thus, $i\in \S$ if and only if $i'\in \S$.

  To prove the density condition, first assume that $i\notin \R$.
  As \eqref{obs:qr} yields $[i\dd i+2\tau)\sm \Q \ne \emptyset$, we can choose a position $j$ in the latter set with minimum identifier.
  If $j < i+\tau$, then $j\in \S$ due to $\id(j) = \min\{\id(j') : j'\in [j\dd j+\tau]\sm \Q\}$.
  Otherwise, $j-\tau \in \S$ due to $\id(j)=\min\{\id(j') : j'\in [j-\tau\dd j]\sm \Q\}$.
  In either case, $\S\cap [i\dd i+\tau) \ne \emptyset$.

  The converse implication is easy: if $i\in \R$, then $[i\dd i+2\tau)\sub \Q$ by \eqref{obs:qr},
  so $[j \dd j+\tau]\sm \Q=\emptyset$ and therefore $j\notin \S$ for $j\in [i\dd i+\tau)$.
  Consequently, $\S\cap[i\dd i+\tau)=\emptyset$.
\end{proof}

The main challenge in building a $\tau$-synchronizing set with \cref{cons:sync} is to choose an appropriate identifier function $\id$ 
so that resulting synchronizing set $\S$ is small. 
We first consider only texts satisfying $\Q=\emptyset$. Note that this is a stronger assumption that $\R=\emptyset$,
which we made in the nonperiodic case of \cref{sec:sort-s,sec:lce,sec:bwt}.

\subsection{The Nonperiodic Case}
\label{sec:construction-of-set-S-simple-case}

The key feature of strings without small periods is that their occurrences cannot overlap
too much.  We say that a set $A \subseteq
\integ$ is \emph{$d$-sparse} if every two distinct elements $i,i'\in A$ satisfy $|i-i'|>d$.
\begin{fact}\label{fct:sparse}
  An equivalence class $\P \in \PP$ is $\frac13\tau$-sparse if
  $\P\cap \Q=\emptyset$.
\end{fact}
\begin{proof}
  Suppose that positions $i,i'\in \P$ satisfy $i<i'\le i+\frac13\tau$.  
    We have $T[i\dd i+\tau)= T[i'\dd i'+\tau)$, so $\per(T[i\dd i'+\tau))\le i'-i\le\frac13\tau$.  
    In particular, $\per(T[i\dd i+\tau))\le \frac13\tau$, so $i\in \Q$ and $\P\cap \Q \ne \emptyset$.
\end{proof}

\subsubsection{Randomized Construction}
\label{sec:construction-of-set-S-rand}

It turns out that choosing $\id$ uniformly at random leads to
satisfactory results if $\Q = \emptyset$.

\begin{fact}\label{fct:rand0}
  Let $\pi:\PP\to[0\dd |\PP|)$ be a uniformly random bijection,
  and let $\id$ be an identifier function such that
  $\id(j)=\pi(\P)$ for each $j\in \P$ and $\P\in \PP$.
  If $\Q=\emptyset$, then a $\tau$-synchronizing set $\S$ defined with \cref{cons:sync}
  based on such a function $\id$ satisfies $\Exp[|\S|]\le \frac{6n}{\tau}$.
\end{fact}
\begin{proof}
  Observe that for every position $i\in [1\dd n-2\tau+1]$,
  we have $|[i\dd i+\tau]\sm \Q| = |[i\dd i+\tau]| = \tau+1$.
  Moreover, \cref{fct:sparse} guarantees that
  $|[i\dd i+\tau]\cap \P|\le 3$ for each class $\P\in \PP$.  Thus,
  positions in $[i\dd i+\tau]$ belong to at least $\frac{\tau}{3}$
  distinct classes.  Each of them has the same probability of
  having the smallest identifier, so 
  \[\Pr[i\in \S] = \Pr[\min\{\id(j) : j \in [i\dd i+\tau]\}\in \{\id(i),\id(i+\tau)\}] \le 2\cdot \tfrac{3}{\tau}\]
  holds for every $i\in [1\dd n-2\tau+1]$.
  By linearity of expectation, we conclude that $\Exp[|\S|] \le \tfrac{6n}{\tau}$.
  \end{proof}

\subsubsection{Deterministic Construction}
\label{sec:construction-of-set-S-det}

Our next goal is to provide an $\Oh(n)$-time deterministic construction of a synchronizing set $\S$
of size $|\S|=\Oh(\frac{n}{\tau})$.
The idea is to gradually build an identifier function $\id$ assigning consecutive identifiers
to classes $\P\in \PP$ one at a time. Our choice of the subsequent classes is guided by a carefully
designed scoring function inspired by \cref{fct:rand0}.

\begin{proposition}\label{lem:sync0}
  If $\Q=\emptyset$ for a text $T$ and a positive integer $\tau\le\frac12 n$,
  then in $\Oh(n)$ time one can construct a $\tau$-synchronizing set of size at most $\frac{18n}{\tau}$.
\end{proposition}
  \begin{proof}
    First, we build the partition $\PP$. A simple $\Oh(n)$-time implementation is based on the suffix array $\SA$
    and the $\LCP$ table of $T$ (see \cref{sec:sa}).
    We cut the suffix array before every position $i$ such that $\LCP[i]<\tau$, and we remove positions $i$ with $\SA[i]>n-\tau+1$.
    The values $\SA[\ell \dd r]$ in each remaining maximal region form a single class $\P \in \PP$.
    We store pointers to $\P$ for all positions $j\in \P$.

    Next, we iteratively construct the function $\id$ represented in a table $\id[1\dd n-\tau+1]$.
    Initially, each value $\id[j]$ is undefined ($\bot$). In the $k$th iteration, we choose a class $\P_k\in \PP$ and
    set $\id[j]=k$ for $j\in \P_k$. Finally, we build a $\tau$-synchronizing set as in \cref{cons:sync}.
    
    Our choice of subsequent classes $\P_k$ is guided by a scoring function.
    To define it, we distinguish \emph{active blocks} which are maximal blocks $\id[\ell\dd r]$ 
    consisting of undefined values only and satisfying $r-\ell \ge \tau$.
    Each position $j\in [\ell\dd r]$ lying in an active block $\id[\ell\dd r]$ is called an \emph{active position}
    and assigned a score: $-1$ if it is among the leftmost or the rightmost $\floor{\frac13 \tau}$ positions
    within an active block
    and $+2$ otherwise.
    Note that the condition $r-\ell\ge \tau$ implies that the total score of positions within an active block
    is non-negative. Hence, the global score of all active positions is also non-negative.

    Our algorithm explicitly remembers whether each position is active and what its score is.
    Moreover, we store the aggregate score of active positions in each class $\P \in \PP$
    and maintain a collection $\PP^+\sub \PP$ of unprocessed classes with non-negative scores.
    Note that the aggregate score is 0 for the already processed classes, so $\PP^+$ is non-empty
    until there are no unprocessed classes.
    Hence, we can always choose the subsequent class $\P_k$ from $\PP^+$.

    Having chosen the class $\P_k$, we need to set $\id[j]=k$ for $j\in \P_k$.
    If the position $j\in \P_k$ was active, then some nearby positions $j'$ may cease to be active or change their scores.
    This further modifies the aggregate scores of other classes $\P\in \PP$, some of which may enter or leave $\PP^+$.
    Nevertheless, the affected positions $j'$ satisfy $|j-j'|\le \tau$, so we can implement the changes
    caused by setting $\id[j]=k$ in $\Oh(\tau)$ time. The total cost of processing $\P_k$
    is $\Oh(|\P_k|+\tau|\A_k|)$, where $\A_k\sub \P_k$ consists of positions which were active initially.

    Finally, we build the synchronizing set $\S$ according to \cref{cons:sync}.
    Recall that a position $i\in [1\dd n-2\tau+1]$ is inserted to $\S$ if and only if $\min \id[i\dd i+\tau] \in \{\id[i],\id[i+\tau]\}$.
    Hence, it suffices to slide a window of fixed width $\tau+1$ over the table $\id[1\dd n-\tau+1]$ computing the sliding-window minima.
    This takes $\Oh(n)$ time.

    To bound the size $|\S|$, consider a position $i$ inserted to $\S$, let $k=\min \id[i\dd i+\tau]$, and note that $k=\id[i]$ or $k=\id[i+\tau]$.
    Observe that prior to processing $\P_k$, we had $\id[j]=\bot$ for $j\in [i\dd i+\tau]$.
    Consequently, $\id[i\dd i+\tau]$ was contained in an active block and $i,i+\tau$ were active positions.
    At least one of them belongs to $\P_k$, so it also belongs to $\A_k$.
    Hence, if $i\in \S$, then $i\in \bigcup_k \A_k$ or $i+\tau \in \bigcup_k \A_k$,
    which yields $|\S| \le 2|\bigcup_k \A_k| = 2\sum_k |\A_k|$.
    In order to bound this quantity, let us introduce sets $\A^+_k\sub \A_k$ formed by active positions that had score $+2$ prior to processing $\P_k$.
    The choice of $\P_k$ as a class with non-negative aggregate score guarantees that $|\A_k|\le 3|\A^+_k|$.

    Finally, we shall prove that the set $\A^+ := \bigcup_k \A^+_k$ is $\frac13\tau$-sparse.  
    Consider two distinct positions $j,j'\in \A^+$ such that $j \in \A^+_k$, $j'\in \A^+_{k'}$.
    If $k=k'$, then $|j-j'|>\frac13\tau$  because $\A^+_k\sub \P_k$ is $\frac13\tau$-sparse by \cref{fct:sparse}.
    Hence, we may assume without loss of generality that $k'<k$.  
    Prior to processing $\P_k$, the position $j'$ was not active (the value $\id[j']=k'$ was already set) whereas $j$ was active and had score
    $+2$.  Hence, there must have been at least $\floor{\frac13\tau}$
    active positions with score $-1$ in between, so $|j-j'|\ge \floor{\frac13\tau} + 1 > \frac13\tau$.  

  Consequently, $|\A^+|\le \big\lceil{\frac{3(n-\tau+1)}{\tau}\big\rceil}\le \frac{3n}{\tau}$, which implies $|\S|\le 2\sum_k |\A_k| \le 6\sum_k |\A^+_k| =  6|\A^+| \le \frac{18n}{\tau}$.
  Moreover, the overall running time is $\Oh(n+\tau\sum_k|\A_k|)=\Oh(n+\tau\frac{n}{\tau}) =\Oh(n)$.
   \end{proof}

\subsubsection{Efficient Implementation for Small \texorpdfstring{$\tau$}{Tau}}
\label{sec:nonperiodic-small-tau}
Next, we shall implement our construction in
$\Oh(\frac{n}{\tau})$ time if $\tau \le \eps \log_{\sigma} n$ and $\eps < \frac15$.

Our approach relies on local consistency of the procedure in \cref{lem:sync0}. 
More specifically, we note that the way
this procedure handles a position $i$ depends only on the
classes of the nearby positions $j$ with $|j-i|\le \tau$.  In
particular, these classes determine how the score of $i$ changes and
 whether $i$ is inserted to $\S$.

Motivated by this observation, we partition $[1\dd n-\tau+1]$ into $\Oh(\frac{n}{\tau})$ \emph{blocks} so that the $b$th block contains
positions $i$ with $\lceil{\frac{i}{\tau}\rceil}=b$.  We
define $T[1+(b-2)\tau\dotdot (b+2)\tau]$ to be the \emph{context} of the $b$th block
 (assuming that $T[i]=\# \notin \Sigma$ for out-of-bound positions $i\notin[1\dd n]$), and we say that blocks are
\emph{equivalent} if they share the same context.

Based on the initial observation, we note that if two blocks $b,b'$ are
equivalent, then the corresponding positions $i=b\tau-\delta$ and $i'=b'\tau-\delta$ (for $\delta\in [0\dd \tau)$)
are processed in the same way by the procedure of \cref{lem:sync0}.  This
means that we need to process just one
\emph{representative block} in each equivalence class.

We can retrieve each context in $\Oh(1)$ time using \cref{prop:packed}.
Consequently, it takes $\Oh(\frac{n}{\tau})$ time to partition the
blocks into equivalence classes and construct a family $\BB$ of
representative blocks. Our choice of $\tau$ guarantees
that $|\BB|=\Oh(1+\sigma^{4\tau})=\Oh(n^{4\eps})$.  Similarly, the
class $\P\in \PP$ of a position $j$ is determined by
$T[j\dd j+\tau)$, so $|\PP|=\Oh(n^\eps)$ and the substring can also be retrieved in $\Oh(1)$
time.  Hence, the construction procedure has $\Oh(n^\eps)$
iterations. If we spent $\Oh\big(\tau^{\Oh(1)}\big)$ time for each
representative block at each iteration, the overall
running time would be $\Oh\big(n^\eps n^{4\eps}\tau^{\Oh(1)}\big)=\Oh\big(n^{5\eps + o(1)}\big)=\Oh\big(n^{1-\Omega(1)}\big)=\Oh(\frac{n}{\tau})$.  
This allows for a
quite simple approach.

We maintain classes $\P\in \PP$ indexed by the underlying
substrings. For each class, we store the identifier $\id(j)$ assigned to
the positions $j\in \P$ and a list of positions $j\in \P$ contained
in the representative blocks.  To initialize these components (with
the identifier $\id(j)$ set to $\bot$ at first), we scan all
representative blocks in $\Oh\big(\tau^{\Oh(1)}\big)$ time per block,
which yields $\Oh\big(|\BB|\tau^{\Oh(1)}\big)=o(\frac{n}{\tau})$ time in
total.

Choosing a class to be processed in every iteration involves computing
scores.  To determine the score of a particular class $\P\in
\PP$, we iterate over all positions $j\in \P$
contained in the representative blocks. We retrieve the class of each
fragment $j'$ with $|j-j'|\le \tau$ in order to compute the
score of $j$. We add this score, multiplied by the number of
equivalent blocks, to the aggregate score of $\P$.  Having computed the
score of each class, we take an arbitrary class $\P_k$ with a
non-negative score (and no value assigned yet), and we assign the
subsequent identifier $k$ to this class.  As announced above, the running
time of a single iteration is $\Oh\big(|\BB|\tau^{\Oh(1)}\big)$ as we spend
$\Oh\big(\tau^{\Oh(1)}\big)$ time for each position contained
in a representative block.

In the post-processing, we compute $\S$ restricted to positions
in representative blocks:  For every
position $i$ contained in a representative block, we retrieve the
classes of the nearby positions $j\in [i\dd i+\tau]$ to
check whether $i$ should be inserted to $\S$.  This~takes $\Oh\big(\tau^{\Oh(1)}\big)$ time per
position $i$, which is $\Oh\big(|\BB|\tau^{\Oh(1)}\big)=o(\frac{n}{\tau})$
in total.

Finally, we build the whole set $\S$: For each block, we copy the
positions of the corresponding representative block inserted to $\S$ (shifting
the indices accordingly).  The running time of this final phase is
$\Oh(|\S|+\frac{n}{\tau})$, which is $\Oh(\frac{n}{\tau})$
due to $|\S|\le \frac{18n}{\tau}$.

\begin{proposition}\label{lem:sync_fast0}
    For every constant $\eps<\frac{1}{5}$, given the packed representation of a
    text $T\in[0\dd\sigma)^n$ and a positive integer $\tau \le \eps \log_\sigma n$
    such that $\Q=\emptyset$, one
    can construct in $\bigO(\frac{n}{\tau})$ time a
    $\tau$-synchronizing set of size $\bigO(\frac{n}{\tau})$.
\end{proposition}

\subsection{The General Case}
\label{sec:arb}

In this section, we adapt our constructions so that they work for arbitrary strings.
For this, we first study the structure of the set $\Q$.

\subsubsection{Structure of Highly Periodic Fragments}
\label{sec:hp}

The probabilistic argument in the proof of
\cref{fct:rand0} relies on the large size of each set $[i\dd i+\tau]\sm \Q$ that we had due to $\Q=\emptyset$. 
However, in general the sets $[i\dd i+\tau]\sm \Q$ can be of arbitrary size between $0$ and $\tau+1$.
To deal with this issue, we define the following set (assuming $\tau\ge 2$).
\begin{camera}
\begin{multline*}
  \B = \big\{i \in [1\dd n-\tau+1]\sm \Q : \\ \per(T[i\dd i+\tau-1)) \le
  \tfrac13\tau\text{ or }\per(T[i+1\dd i+\tau)) \le \tfrac13\tau\big\}
\end{multline*}
\end{camera}
\begin{full}
   \[ \B = \big\{i \in [1\dd n-\tau+1]\sm \Q : \per(T[i\dd i+\tau-1)) \le
    \tfrac13\tau\text{ or }\per(T[i+1\dd i+\tau)) \le \tfrac13\tau\big\}\]
\end{full}
(In the special case of $\tau=1$, we set $\B=\emptyset$.)
Intuitively, $\B$ forms a boundary which separates $\Q$ from
its complement, as formalized in the fact
below. However, it also contains some additional fragments included to
make sure that $\B$ consists of full classes $\P\in \PP$.

\begin{camera}
  \begin{fact}[\fullonly]\label{fct:boundary}
    If $[\ell\dd r]\cap \Q\ne \emptyset$ and $[\ell\dd r]\not\subseteq\Q$ for
    two positions $\ell,r\in [1\dd n-\tau+1]$, then $[\ell\dd r]\cap \B\ne \emptyset$.
  \end{fact}
\end{camera}
\begin{full}
  \begin{fact}\label{fct:boundary}
    If $[\ell\dd r]\cap \Q\ne \emptyset$ and $[\ell\dd r]\not\subseteq\Q$ for
    two positions $\ell,r\in [1\dd n-\tau+1]$, then $[\ell\dd r]\cap \B\ne \emptyset$.
  \end{fact}
\begin{proof}
  We proceed by induction on $r-\ell$.  In the base case, $r=\ell+1$
  and the assumption yields that $\{\ell,\ell+1\}\sm \Q$ consists of a single element $i$.
  However, this means that $i-1\in \Q$ (if $i=\ell+1$) or $i+1\in \Q$ (if $i=\ell$).
  We conclude that $\per(T[i\dd i+\tau-1)) \le
  \frac13\tau$ or $\per(T[i+1\dd i+\tau))\le \frac13\tau$, respectively, so $i\in \B$ as claimed.

  For the inductive step with $r-\ell \ge 2$, it suffices to note that if $[\ell\dd r]\cap \Q\ne
  \emptyset$ and $[\ell\dd r]\not\subseteq\Q$, 
  then the same is true for $[\ell\dd r-1]$ or for $[\ell+1\dd r]$
  (because these two subsets have a non-empty intersection).
\end{proof}
\end{full}

We conclude the analysis with a linear-time construction of $\Q$ and $\B$,
which also reveals an upper bound on $|\B|$.

\begin{camera}
\begin{lemma}[\fullonly]\label{lem:hb}
  Given a text $T$ and a positive integer $\tau$, the sets $\Q$ and $\B$ 
  can be constructed in $\Oh(n)$ time. Moreover,
  $|\B|\le \frac{6n}{\tau}$.
\end{lemma}
\end{camera}
\begin{full}
  \begin{lemma}\label{lem:hb}
    Given a text $T$ and a positive integer $\tau$, the sets $\Q$ and $\B$ 
    can be constructed in $\Oh(n)$ time. Moreover,
    $|\B|\le \frac{6n}{\tau}$.
  \end{lemma}
\begin{proof}
  Note that if $\tau\le 2$, then $\B=\Q=\emptyset$, and the claim holds trivially.
  We assume $\tau\ge 3$ henceforth.

  Let $I=[i\dd i+b)\sub [1\dd n-\tau+1]$ be a block of $b \le \ceil{\frac13\tau}$
  subsequent positions. Consider a fragment $x = T[i+b\dd i+\tau-1)$ (non-empty due to $\tau\ge 3$), and let $p = \per(x)$.
  Let us further define $y=T[\ell\dd r]$ as the maximal fragment with period $p$ containing $x$ and contained in $T[i\dd i+b+\tau-1)$.
  We claim that $I\cap \Q = I \cap \B = \emptyset$ if $p > \frac13\tau$ or $|y|<\tau-1$,
  whereas $I\cap \Q = I \cap [\ell\dd r-\tau+1]$ and $I\cap \B = I \cap \{\ell-1, r-\tau+2\}$ otherwise.

  First, let us consider a position $j\in I\cap [\ell \dd r-\tau+1]$ provided that $p\le \frac13\tau$.
  Observe that $T[j\dd j+\tau)$ is contained in $y$, so $\per(T[j\dd j+\tau))=p \le \frac13\tau$,
  which implies $j\in I\cap \Q$.
  On the other hand, if $j\in I \cap \Q$, we define $p' = \per(T[j \dd j+\tau))$.
  Note that $x$ is contained in $T[j \dd j+\tau)$ and thus also has $p'\le \frac13\tau$ as its period.
  Moreover, $p+p'-1 \le \floor{\frac23\tau}-1 = \tau-1-\ceil{\frac13\tau} \le \tau-1-b = |x|$, so $p \mid p'$  by the periodicity lemma (\cref{lem:per}).
  Consequently, $p=p'\le \frac13\tau$ and $T[j\dd j+\tau)$ is contained in $y$, which yields $j\in I\cap [\ell \dd r-\tau+1]$.

  Next, let us consider a position $j\in I \cap \{\ell-1,r-\tau+2\}$ provided that $p\le \frac13\tau$ and $|y|\ge \tau-1$.
  Observe that $T[j+1\dd j+\tau)$ or $T[j\dd j+\tau-1)$ is contained in $y$,
  so $\per(T[j+1\dd j+\tau))=p \le \frac13\tau$ or $\per(T[j\dd j+\tau-1))=p \le \frac13\tau$.
  Furthermore, $j\notin \Q$ due to $j \notin I\cap [\ell\dd r-\tau+1]=I\cap \Q$, so $j \in I\cap \B$.
  On the other hand, if $j\in I\cap \B$, we define $p' = \per(T[j+1 \dd j+\tau-1))$.
  Note that $x$ is contained in $T[j+1 \dd j+\tau-1)$ and thus also has $p'\le \frac13\tau$ as its period.
  Moreover, $p+p'-1 \le \floor{\frac23\tau}-1 = \tau-1-\ceil{\frac13\tau} \le \tau-1-b = |x|$, so $p \mid p'$ by the periodicity lemma (\cref{lem:per}).
  Consequently,  $p=p'\le \frac13\tau$ and $T[j+1\dd j+\tau-1)$ is contained in $y$, which yields $j\in I\cap [\ell-1 \dd r-\tau+2]$ and $|y|\ge \tau-1$.
  Furthermore, $j\notin I\cap \Q =  I\cap [\ell\dd r-\tau+1]$, so $j\in I\cap \{\ell-1,r-\tau+2\}$ holds as claimed.

  Finally, we observe that $p=\per(x)$ can be computed in $\Oh(|x|)=\Oh(\tau)$ time~\cite{morris1970linear,DBLP:journals/siamcomp/KnuthMP77},
  and $y$ can be easily constructed in $\Oh(|y|)=\Oh(\tau)$ time as a greedy extension of $x$.
  Hence, it takes $\Oh(\tau)$ time to determine $I\cap \Q$ and $I\cap \B$. Moreover, the size of the latter set is at most two.

  The domain $[1\dd n-\tau+1]$ can be decomposed into $\big\lceil{\frac{n-\tau+1}{\lceil{\tau/3}\rceil}\big\rceil}\le \frac{3n}{\tau}$ blocks $I$ of size $b\le \ceil{\frac13\tau}$, so we conclude that the sets $\Q$ and $\B$ can be constructed in $\Oh(n)$ time and that the size of the latter set is at most $\frac{6n}{\tau}$.\end{proof}
\end{full}

\subsubsection{Randomized Construction}
\label{sec:arb-rand}

The set $\B$ lets us adopt
the results of \cref{sec:construction-of-set-S-simple-case} to arbitrary
texts. As indicated in a probabilistic argument, the key trick is to assign the smallest identifiers to classes~in~$\B$.  

\begin{fact}\label{lem:rand}
  There is an identifier function $\id$ such that \cref{cons:sync} yields
  a synchronizing set of size at most~$\frac{18n}{\tau}$.
\end{fact}
\begin{proof}
  As in \cref{fct:rand0}, we take a random bijection $\pi :\PP \to
  [0\dd |\PP|)$ and set $\id(j)=\pi(\P)$ if $j\in \P$.
  However, this time we draw $\pi$ uniformly at random only among
  bijections such that if $\P\subseteq \B$ and $\P'\cap
  \B=\emptyset$ for classes $\P,\P'\in \PP$, then
  $\pi(\P)<\pi(\P')$.  (Note that each class in $\PP$ is either
  contained in $\B$ or is disjoint with this set.)

  Consider a position $i$. Observe that if $[i\dd i+\tau]\cap \B\ne \emptyset$, then $i \in \S$ holds only 
  if $i\in \B$ or $i+\tau\in \B$.
  Hence, the number of such positions $i\in \S$ is at most $2|\B|\le \frac{12n}{\tau}$.
  Otherwise, \cref{fct:boundary} implies that $[i\dd i+\tau]\subseteq\Q$ or $[i\dd i+\tau]\cap \Q=\emptyset$.
  In the former case, we are guaranteed that $i\notin \S$ by \cref{cons:sync}.
  On the other hand, $\Pr[i \in \S]\le \frac{6}{\tau}$ holds in the latter case as in the proof of \cref{fct:rand0}
  since $[i\dd i+\tau]\sm \Q = [i\dd i+\tau]$ is of size $\tau+1$.
  By linearity of expectation, the expected number of such positions $i\in \S$ is up to~$\frac{6n}{\tau}$.

  We conclude that $\Exp[|\S|]\le \frac{12n}{\tau}+\frac{6n}{\tau}=\frac{18n}{\tau}$.
  In particular, $|\S|\le
  \frac{18n}{\tau}$ holds for some identifier
  function $\id$.
\end{proof}

\subsubsection{Deterministic Construction}
\label{sec:arb-det}

Our adaptation of the deterministic construction uses
\cref{fct:boundary,lem:hb} in a similar way.  As in
the proof of \cref{lem:sync0}, we gradually construct $\id$ handling one partition class $\P\in \PP$ at a time.
We start with classes contained
in $\B$ (in an arbitrary order), then we process  classes contained in $\Q$ (still in an arbitrary order).
In the final third phase, we choose the subsequent classes disjoint with $\B\cup \Q$ according to their scores.

\begin{camera}
\begin{proposition}[\fullonly]\label{lem:sync}
  Given a text $T\in [0\dd \sigma)^n$ for $\sigma = n^{\Oh(1)}$ and a positive integer $\tau\le\frac12 n$,
  in $\Oh(n)$ time one can construct a $\tau$-synchronizing set of size at most $\frac{30n}{\tau}$.
\end{proposition}
\end{camera}
\begin{full}
  \begin{proposition}\label{lem:sync}
    Given a text $T\in [0\dd \sigma)^n$ for $\sigma = n^{\Oh(1)}$ and a positive integer $\tau\le\frac12 n$,
    in $\Oh(n)$ time one can construct a $\tau$-synchronizing set of size at most $\frac{30n}{\tau}$.
  \end{proposition}
  \begin{proof}
    First, we build the partition $\PP$ as in the proof of \cref{lem:sync0}.
    Next, we construct $\Q,\B\sub[1\dd n-\tau+1]$ using \cref{lem:hb}
    and identify each class $\P\in \PP$ as contained in $\B$, contained in $\Q$, or disjoint with $\B\cup \Q$.

    We start the iterative construction of an identifier function $\id$ by initializing a table
    $\id[1\dd n-\tau+1]$ with markers $\bot$ representing undefined values.
    In the first two phases, we process the classes contained in $\B$ and $\Q$, respectively, assigning them initial subsequent
    identifiers and setting $\id[j]=k$ for the $k$th class $\P_k \in \PP$ considered.

    Before moving on to classes disjoint with $\B\cup \Q$, we compute the auxiliary components required for the scoring function
    (defined exactly as in the proof of \cref{lem:sync0}).
    For this, we scan the table $\id$ to identify
    active blocks and assign scores to active positions, maintaining the aggregate score of each class $\P\in \PP$
    and the collection $\PP^+$ of unprocessed classes with non-negative score.
    The total running time up to this point is clearly $\Oh(n)$.

    The subsequent classes $\P_k$ are processed exactly as in the proof of \cref{lem:sync0}:
    we choose the class $\P_k$ from $\PP^+$ and set $\id[j]=k$ for each $j\in \P_k$, updating the active positions, scores, aggregate scores, and $\PP^+$ accordingly.
    The running time of a single iteration is still $\Oh(|\P_k|+\tau|\A_k|)$, where $\A_k\sub \P_k$ consists of positions active
    prior to processing $\P_k$. 
    The original argument from the proof of \cref{lem:sync0} still shows that $\sum_{k}|\A_k| \le \frac{9n}{\tau}$
    (this is because $\A_k \sm \Q= \emptyset$ for each $k$), so the running time of the third phase is 
     $\Oh(\sum_k |\P_k| + \tau\sum_{k}|\A_k|)=\Oh(n)$.

    Finally, we build the synchronizing set $\S$ according to \cref{cons:sync} using the sliding-window approach from the proof of \cref{lem:sync0}.
    The only difference is that we have to ignore identifiers $\id[j]$ of positions $j\in \Q$ while computing the sliding-window minima.

    It remains to bounds to size of the set $\S$ constructed this way.
    For this, we shall prove that if $i\in \S$, then $i$ or $i+\tau$ belongs to $\B\cup \bigcup_{k}\A_k$,
    from which we conclude the desired bound: $|\S| \le 2|\B|+ 2\sum_k|\A_k| \le \frac{12n}{\tau}+\frac{18n}{\tau}$.

    Hence, let us consider a position $i\in \S$ and recall that $\min \{\id[j]: j\in [i\dd i+\tau]\sm \Q\} \in\{\id[i],\id[i+\tau]\}$
    by \cref{cons:sync}.
    Since we processed the classes contained in $\B$ first,
    if $[i\dd i+\tau]\cap \B \ne \emptyset$, then the minimum on the left-hand side is attained by $\id[j]$ with $j\in \B$.
    Consequently, $i\in \B$ or $i+\tau\in \B$, consistently with the claim.

    Otherwise, \cref{fct:boundary} yields that $[i\dd i+\tau]\cap \Q = \emptyset$ or $[i\dd i+\tau]\sub \Q$,
    with the latter case infeasible due to $i\in \S$.
    Consequently, $[i\dd i+\tau]\sm \Q=[i\dd i+\tau]$ and we observe
    that $k=\min \id[i\dd i+\tau]$ satisfies $k\in \{\id[i], \id[i+\tau]\}$.
    Hence, we had $\id[j]=\bot$ for $j\in [i\dd i+\tau]$ prior to processing the class $\P_k$.
    This class is disjoint with $\B\cup \Q$ and each of these positions $j$ was active at that point.
    In particular, $i$ and $i+\tau$ were active, so $i\in \A_k$ or $i+\tau\in \A_k$, consistently with the claim.

    This completes the characterization of positions $i\in \S$ resulting in $|\S|\le \frac{30n}{\tau}$. 
   \end{proof}
\end{full}

\subsubsection{Efficient Implementation for Small \texorpdfstring{$\tau$}{Tau}}
\label{sec:arb-small-tau}

We conclude by noting that the procedure of \cref{lem:sync}
can be implemented in $\Oh(\frac{n}{\tau})$ time for $\tau < \eps \log_{\sigma} n$ 
just as we implemented the procedure of \cref{lem:sync0} to prove \cref{lem:sync_fast0}.
The only observation needed to make this seamless adaptation is that we can check
in $\Oh\big(\tau^{\Oh(1)}\big)$ time whether a given position belongs to $\Q$ or $\B$.
In particular, we have sufficient time to perform these two checks for every position contained
in a representative block or its context.

\begin{theorem}\label{lem:sync_fast}
  For every constant $\eps<\frac{1}{5}$, given the packed representation of a
  text $T\in[0\dd\sigma)^n$ and a positive integer $\tau \le \eps \log_\sigma n$, one
  can construct in $\bigO(\frac{n}{\tau})$ time a
  $\tau$-synchronizing set of size $\bigO(\frac{n}{\tau})$.
\end{theorem}

\bibliographystyle{plainurl}
\bibliography{stoc2019}

\end{document}